\documentclass{llncs}

\let\doendproof\endproof
\renewcommand\endproof{~\hfill$\qed$\doendproof}

\usepackage[bibliography=common]{apxproof}
\usepackage{amsfonts,amsmath,amssymb,amsthm}
\usepackage{paralist}
\usepackage{hyperref}
\usepackage{graphicx}
\usepackage{subfigure}

\usepackage{todonotes}
\usepackage{xspace}
\usepackage{complexity}

\usepackage[capitalise]{cleveref}
\Crefname{observation}{Observation}{Observations}
\Crefname{algorithm}{Algorithm}{Algorithms}
\Crefname{section}{Sect.}{Sects.}
\Crefname{observation}{Observation}{Observations}
\Crefname{lemma}{Lemma}{Lemmas}
\Crefname{lemma2}{Lemma}{Lemmas}
\Crefname{theorem2}{Theorem}{Theorems}
\Crefname{claim}{Claim}{Claims}
\Crefname{claimx}{Claim}{Claims}
\Crefname{figure}{Fig.}{Figs.}
\Crefname{enumi}{Condition}{Conditions}
\Crefname{property}{Property}{Properties}
\Crefname{assumption}{Assumption}{Assumptions}

\theoremstyle{definition}
\newtheorem{defn}{Definition}
\newtheorem{clm}{Claim}

\newtheoremrep{theorem2}[theorem]{Theorem}
\newtheoremrep{lemma2}[lemma]{Lemma}
\newtheoremrep{observation2}[theorem]{Observation}
\newtheoremrep{property2}[theorem]{Property}
\newtheoremrep{defn2}[defn]{Definition}
\newtheoremrep{clm2}[clm]{Claim}

\newcommand{\qedclaim}{\hfill $\blacksquare$}
\newenvironment{claimproof}{\noindent{\itshape Proof.}}{~\hfill \qedclaim \smallskip}

\begin{document}

%
    %Point-Sets Supporting Graph Stories
    % How to Get Away with Graph Stories
\title{
    Small Point-Sets Supporting Graph Stories\thanks{This work was partially supported by: $(i)$ MIUR, grant 20174LF3T8 ``AHeAD: efficient Algorithms for HArnessing networked Data''; $(ii)$ Dipartimento di Ingegneria - Universit\`a degli Studi di Perugia, grants RICBA20EDG: ``Algoritmi e modelli per la rappresentazione visuale di reti'' and RICBA21LG: ``Algoritmi, modelli e sistemi per la rappresentazione visuale di reti''.}
}
 \author{
     Giuseppe Di~Battista\inst{1}\texorpdfstring{\href{https://orcid.org/0000-0003-4224-1550}{\protect\includegraphics[scale=0.45]{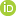}}}{} \and
     Walter Didimo\inst{2}\texorpdfstring{\href{https://orcid.org/0000-0002-4379-6059}{\protect\includegraphics[scale=0.45]{orcid}}}{} \and
     Luca Grilli\inst{2}\texorpdfstring{\href{https://orcid.org/0000-0002-2463-3772}{\protect\includegraphics[scale=0.45]{orcid}}}{} \and
     Fabrizio Grosso\inst{1}\texorpdfstring{\href{https://orcid.org/0000-0002-5766-4567}{\protect\includegraphics[scale=0.45]{orcid}}}{} \and
     Giacomo Ortali\inst{2}\texorpdfstring{\href{https://orcid.org/0000-0002-4481-698X}{\protect\includegraphics[scale=0.45]{orcid}}}{} \and
     Maurizio Patrignani\inst{1}\texorpdfstring{\href{https://orcid.org/0000-0001-9806-7411}{\protect\includegraphics[scale=0.45]{orcid}}}{} \and
     Alessandra Tappini\inst{2}\texorpdfstring{\href{https://orcid.org/0000-0001-9192-2067}{\protect\includegraphics[scale=0.45]{orcid}}}{}
 }
%\author{Anonymized submission}

%\authorrunning{Anonymized submission}
\authorrunning{G. Di~Battista et al.}
%\institute{}
 \institute{
     Roma Tre University, Rome, Italy
     \email{\{giuseppe.dibattista,fabrizio.grosso,maurizio.patrignani\}@uniroma3.it} \and
     University of Perugia, Perugia, Italy
     \email{\{walter.didimo,luca.grilli,giacomo.ortali,alessandra.tappini\}@unipg.it}
 }
\maketitle              % typeset the header of the contribution
%\linenumbers
%
\begin{abstract}
% Graph stories are a framework for exploring temporal data.
% called the size of the viewing window
In a graph story the vertices enter a graph one at a time and each vertex persists in the graph for a fixed amount of time $\omega$, called viewing window. At any time, the user can see only the drawing of the graph induced by the vertices in the viewing window and this determines a sequence of drawings.
For readability, we require that all the drawings of the sequence are planar. For preserving the user's mental map we require that when a vertex or an edge is drawn, it has the same drawing for its entire life. We study the problem of drawing the entire sequence by mapping the vertices only to $\omega+k$ given points, where $k$ is as small as possible.
We show that:
$(i)$ The problem does not depend on the specific set of points but only on its size;
$(ii)$ the problem is \NP-hard and is FPT when parameterized by $\omega+k$;
%
%$(iii)$ there are families of graph stories that can be drawn with $k=0$ for any $\omega$ and families of graph stories that can be drawn (cannot be drawn) with $k=0$ if $\omega$ is small; and
$(iii)$ there are families of graph stories that can be drawn with $k=0$ for any $\omega$, while for $k=0$ and small values~of~$\omega$ there are families of graph stories that can be drawn and others that cannot;
$(iv)$ there are families of graph stories that cannot be drawn for any fixed $k$ and families of graph stories that require at least a certain~$k$.

% We equip a graph story with $\omega+k$ points and study the problem of drawing the entire sequence placing the vertices only in such points, where $k$ is as small as possible. 

\end{abstract}

%%%%%%%%%%%%%%%%%%%
%% Introductions %%
%%%%%%%%%%%%%%%%%%%
\section{Introduction}\label{se:intro}

In this paper we address ``graph stories'', a model introduced by Borrazzo et al.\ in~\cite{bddfp-gsisa-jgaa20} as a framework for exploring temporal data. In a \emph{graph story} the vertices enter a graph one at a time and persist in the graph for a fixed amount of time~$\omega$, called the \emph{size of the viewing window}.
At any time, the user can see only the drawing of the graph induced by the vertices in the viewing window and this determines a sequence of drawings.
For readability, all the drawings of the sequence are required to be planar. For preserving the user's mental map, when an edge is drawn it has the same drawing for its entire life. 
%To not constrain the drawing too much, the edges are drawn as Jordan arcs.
Also, in order to limit the constraints, we allow the edges to be represented by Jordan arcs.

% For readability, we require that all the drawings of the sequence are planar. For preserving the user's mental map we require that when an edge is drawn, it has the same drawing for its entire life. For not constraining the drawing too much, we draw edges as Jordan curves.

Graph stories are related to a rich body of literature devoted to the visualization of dynamic graphs (surveys can be found in~\cite{DBLP:journals/cgf/BeckBDW17,DBLP:journals/im/Michail16}). One of the main classification criteria of dynamic graph problems is whether the story is entirely known in advance ({\em off-line model}) or not ({\em on-line model}). In this respect, our contribution falls in the off-line model. A third intermediate category ({\em look-ahead model}) is when a small chunk of the incoming events is known in advance to the drawing algorithms. The events are also a classification criterion, as they may refer to vertices, edges, or both. Finally, further constraints may regard the timings of the events, the more common being that they occur one at a time at regular intervals and that the incoming objects have a fixed lifetime as in the case of graph stories. In some cases, the order of the events is constrained to correspond to a specific kind of visit of the graph.

Several results focus on dynamic trees. In~\cite{DBLP:journals/ipl/BinucciBBDGPPSZ12}, it is shown how to draw a tree in $O(\omega^3)$ area where the model is on-line, the incoming objects are edges that arrive in the order of a Eulerian tour of the tree and whose straight-line drawing persists for a fixed lifetime~$\omega$. 
In~\cite{DBLP:conf/gd/DemetrescuBFLPP99}, a small look ahead on the sequence of vertices is used in order to add one vertex at a time to the current drawing of an infinite tree, balancing the readability of the drawings with respect to the difference between consecutive drawings. 
In~\cite{DBLP:conf/gd/SkambathT16}, a sequence of trees (their union, though, may be an arbitrary graph) is completely known in advance. Vertices and edges can move during the animation and can have arbitrary lifetime. The purpose is to pursue aesthetic criteria commonly adopted for tree drawings~\cite{DBLP:journals/tse/ReingoldT81}.

Only a few results regard more complex families of graphs.  For instance, in~\cite{DBLP:conf/gd/GoodrichP13a}, a stream of edges enter the drawing and never leave it, forming an outerplanar graph that has to be drawn according to an on-line model, moving the previously drawn vertices by a polylogarithmic distance.
In~\cite{DBLP:journals/siamcomp/CohenBTT95} 
%Dynamic Graph Drawings: Trees, Series-Parallel Digraphs, and Planar ST-Digraphs
the drawings of several families of graphs are updated as vertices and edges enter and leave the current graph according to the on-line model.

More feebly related to our setting is the literature about dynamic planarity \cite{%
DBLP:journals/siamcomp/BattistaT96,%
%On-Line Planarity Testing
DBLP:journals/iandc/BattistaTV01,%
%Incremental Convex Planarity Testing
DBLP:reference/algo/Italiano16e,%
%Fully Dynamic Planarity Testing
DBLP:conf/stoc/Poutre94,%
%Alpha-algorithms for incremental planarity testing
DBLP:conf/gd/RextinH08%
%A Fully Dynamic Algorithm to Test the Upward Planarity of Single-Source Embedded Digraphs
}, where the model is on-line and the planar embedding of the graph is allowed to change. When the embedding has to be preserved, instead, planarly adding a stream of edges with a fixed lifetime is \NP-complete even~for~the off-line model~\cite{DBLP:journals/tcs/LozzoR19}.
%Other results for the topological setting are presented in~\cite{DBLP:journals/algorithmica/AngeliniB19,DBLP:conf/gd/Schaefer14}.  
%
Also, related somehow to dynamic graph drawing is geometric simultaneous embedding~\cite{DBLP:reference/crc/BlasiusKR13,DBLP:journals/comgeo/BrassCDEEIKLM07}, which can be used to model temporal graphs. 

Coming more properly to the graph story model, Borrazzo et al.~\cite{bddfp-gsisa-jgaa20} address the setting where all the drawings of the story are straight-line and planar, and where vertices do not change their position once drawn. 
It is shown that graph stories of paths and trees can be drawn on a $2\omega \times 2\omega$ and on an $(8\omega+1) \times (8\omega+1)$ grid, respectively. Further, there exist graph stories of planar graphs that cannot be drawn straight-line within an area that is only a function of~$\omega$.

\smallskip
\noindent \textbf{Contribution.} We study the problem of drawing a graph story by mapping the vertices only to $\omega+k$ given points, where $k$ is as small as possible. We call this a \emph{realizability problem}.
%
%We present the following results.
Our contribution is as follows.
In Section~\ref{se:graph-stories} we show that the realizability of graph stories is a topological problem, as it does not depend on the specific set of points but only on its size. We also give a characterization of realizable graph stories based on the concept of ``compatible embeddings''.
In Section~\ref{se:realizability-testing} the realizability problem is proven to be \NP-complete, even for any given constant $k$, and to belong to \FPT~when parameterized by the size $\omega+k$ of the point set.
In Section~\ref{se:minimal} we study the realizability of graph stories with $k=0$, which we call \emph{minimal}. In particular, we show that: 
$(i)$~Every minimal graph story of an outerplanar graph is realizable;
$(ii)$~for every $\omega \geq 5$ there exist minimal graph stories of series-parallel graphs that are not realizable; %(discutere il rapporto con il \Cref{th:sp-non-minimal} sui serie-parallelo, diverso valore di $\omega$)
$(iii)$~all minimal graph stories with $\omega \leq 5$ whose graph does not contain $K_5$ are realizable if we are allowed to redraw at most one edge at each vertex arrival; 
%interval of time; 
and $(iv)$~minimal graph stories with $\omega \leq 5$ are always realizable for planar triconnected cubic graphs.
%and their $1$-reroute realizability.
%Every minimal graph story such that $G$ is planar is $O(\omega)$-reroute realizable. (commentare)
%
Finally, in Section \ref{se:lower-bounds} we show that there are families of graph stories that are not realizable for any fixed $k$ and families of graph stories that, to be realizable, require at least a certain value for $k$.

Some proofs have been sketched or omitted and can be found in the appendix.

\medskip\noindent\textbf{Preliminaries.}
%\label{se:preliminaries}
A \emph{drawing} $\Gamma$ of a graph $G=(V,E)$ maps each vertex of $V$ to a distinct point of the plane and each edge of $E$ to a Jordan arc connecting its end-vertices;
$\Gamma$ is \emph{planar} if no two edges intersect except at common endpoints.
%A drawing is \emph{planar} if it contains no vertex-edge overlaps and no edge-edge crossings.
%
A planar drawing $\Gamma$ of $G$ subdivides the plane into connected regions called \emph{faces}, and the set of circular orders of the edges incident to each vertex is called a \emph{rotation system}.
The unbounded face of $\Gamma$ is the~\emph{external~face}.
Walking on the (not necessarily connected) border of a face~$f$ of $\Gamma$ so to keep~$f$ to the left  determines a set, called the \emph{boundary of~$f$}, of circular lists of alternating vertices and edges. Each list describes a (not necessarily simple) cycle, which can also consist of an isolated vertex: 
%Such a set is the \emph{boundary} of $f$.
%Two vertices (resp. edges) in a list of the boundary of $f$ are \emph{consecutive} if there is no other vertices (resp. edges) between them in the list.
%The lists in the boundary of the same face are vertex-disjoint
Each edge of $G$ occurs either once in exactly two circular lists of different face boundaries or twice in the circular list of one face~boundary.

Two planar drawings of $G$ are \emph{equivalent} if they have the same rotation system, face boundaries, and external face.
%and their external faces are bounded by the same circular lists of vertices.
%
An equivalence class of planar drawings of $G$ is a \emph{planar embedding of $G$}.
Note that, if $G$ is connected then each face boundary consists of exactly one circular list; in this case an embedding of $G$ is fully specified by its rotation system and by its external face.
If~$G$ is equipped with a planar embedding $\phi$, it is a \emph{plane} graph; a planar drawing $\Gamma$ of $G$ is \emph{embedding-preserving}~if $\Gamma \in \phi$.
If $G'$ is a subgraph of $G$ and $\Gamma'$ is the restriction of $\Gamma$ to $G'$, the planar embedding $\phi'$ of $\Gamma'$ is the \emph{restriction of $\phi$~to~$G'$}.

\section{Graph Stories}\label{se:graph-stories}
\begin{defn}\label{df:graph-story}
A \emph{graph story} is a tuple ${\cal S} = (G, \omega, k, \tau)$ where:
\begin{inparaenum}[$(i)$]
    \item $G=(V,E)$ is an $n$-vertex graph;
    \item $\omega \leq n$ is a positive integer,
    %such that $\omega \leq n$, 
    called the \emph{size of the viewing window};
    \item $k$ is a non-negative integer, called the \emph{number of extra points}; and
    \item $\tau = \langle v_1, v_2, \dots, v_n \rangle$ is a linear ordering of the vertices of $G$ (i.e., $v_i \in V$ is the vertex at position $i$ according to $\tau$).
\end{inparaenum}
\end{defn}

Let $G_i=(V_i, E_i)$ denote the subgraph of $G$ induced by all vertices $v_j$ such that $ \max\{1,i-\omega+1\} \leq j \leq i$. Observe that, if $i \leq \omega$ then $G_i$ consists of the $i$ vertices $\{v_1, v_2, \dots, v_i\}$; otherwise $G_i$ consists of the $\omega$ vertices $\{v_{i-\omega+1}, v_{i-\omega+2}, \dots, v_i\}$. In other words, $G_i$ is the subgraph induced by $v_i$ and by the (up to) $\omega-1$ vertices of $G$ that precede $v_i$ in $\tau$. For each $i$, we say that $v_i$ \emph{enters} the viewing window at time $i$, and for each $i \in \{\omega+1, \dots, n\}$, we say that $v_{i-\omega}$ \emph{leaves} the viewing window at time $i$.

\begin{defn}\label{df:realization}
A \emph{realization} of a graph story ${\cal S} = (G, \omega, k, \tau)$ \emph{on a set} $P$ of $\omega+k$ points is a sequence of drawings ${\cal R} = \langle \Gamma_1, \Gamma_2, \dots, \Gamma_n \rangle$ with the following two properties:
%\begin{itemize}
(\textsf{R1}) $\Gamma_i$ $(1 \leq i \leq n)$ is a planar drawing of $G_i$, where distinct vertices of $V_i$ are mapped to distinct points of $P$;
(\textsf{R2}) the restrictions of $\Gamma_{i-1}$ and of $\Gamma_i$ $(2 \leq i \leq n)$ to their common subgraph $G_{i-1} \cap G_i$ are identical. 
%\end{itemize}

\begin{figure}[tb]
	\centering
	\subfigure[$G$]{\label{fi:story-example-a}\includegraphics[page=1,width=0.193\columnwidth]{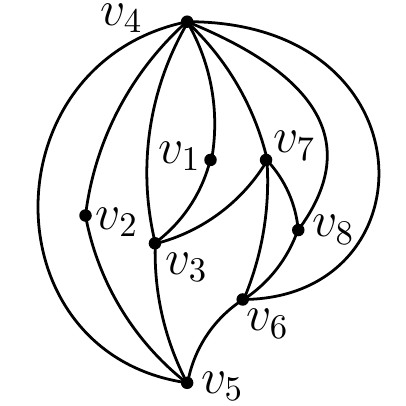}}
	\hfill
	\subfigure[$\Gamma_5$]{\label{fi:story-example-b}\includegraphics[page=2,width=0.193\columnwidth]{story-example}}
	\hfill
	\subfigure[$\Gamma_6$]{\label{fi:story-example-c}\includegraphics[page=3,width=0.193\columnwidth]{story-example}}
	\hfill
	\subfigure[$\Gamma_7$]{\label{fi:story-example-d}\includegraphics[page=4,width=0.193\columnwidth]{story-example}}
	\hfill
	\subfigure[$\Gamma_8$]{\label{fi:story-example-e}\includegraphics[page=5,width=0.193\columnwidth]{story-example}}
	\caption{A realization of a graph story ${\cal S} = (G, 5, 1, \tau)$ on a set $P$ of 6 points. The points of $P$ are yellow disks. For each $\Gamma_i$ ($5 \le i \le 8$), vertex $v_i$ and its incident edges are red.}\label{fi:story-example}
\end{figure}

\medskip \cref{fi:story-example} shows a realization of a graph story ${\cal S} = (G, 5, 1, \tau)$ on a set of 6 points.
%, where vertices are denoted with their subscript in $\tau = \langle v_1, v_2, \dots, v_8 \rangle$.
A graph story ${\cal S}$ is \emph{realizable} if there exists a set $P$ of $\omega+k$ points~such that~$\cal S$ admits a realization on $P$. Since the planarity of all graphs $G_i$ is necessary for realizability, from now on we consider graph stories that satisfy~this~requirement.
\end{defn}

\begin{remark}[Edge Visibility]\label{re:visibility}
We assume that $G$ only consists of \emph{visible edges}, i.e., edges $(v_i, v_j)$ such that $|i-j| < \omega$. Indeed if $|i-j| \geq \omega$, $(v_i, v_j)$ can be ignored, as it never appears in a realization.   
%Any edge $(v_i, v_j)$ of $G$ such that $|i-j| \geq \omega$ can be ignored as it never appears in a realization of the story. Hence, we assume that $G$ only consists of \emph{visible edges}, i.e., for every edge $(v_i, v_j)$ of $G$ we have $|i-j| < \omega$.
%We assume that $G$ only consists of \emph{visible edges}, i.e., for every edge $(v_i, v_j)$ of $G$ we have $|i-j| < \omega$ (the other edges never appear in the story).
%
Our assumption has two implications:
%\begin{itemize}
$(i)$~$G$ has vertex-degree at most $2\omega-2$ (every~$G_i$ has vertex-degree at most $\omega-1$); and
$(ii)$ $G$ has {\em bandwidth} at most $\omega-1$ and hence {\em pathwidth} at most $\omega-1$ \cite{doi:10.1137/S0097539793258143} (the set of bags of this decomposition is $\{V_1, V_2, \dots, V_n\}$).
%\end{itemize}
\end{remark}

\begin{remark}[Minimality]\label{re:minimality}
Clearly, if a graph story ${\cal S}=(G,\omega,k,\tau)$ is realizable, every other story ${\cal S'}=(G,\omega,k',\tau)$ with $k'>k$ is realizable too. Hence, a natural scenario is when the number of extra points $k$ is zero.
%, i.e., when the story must be realized on a set of points whose size equals $\omega$.
We call such~a~story \emph{minimal} and we denote it as ${\cal S} = (G,\omega,\tau)$.
For a minimal graph~story, Property~\textsf{R2} of \cref{df:realization} implies that each vertex $v_i$ with $\omega+1 \leq i \leq n$ is mapped to the same point as $v_{i-\omega}$, thus the mapping of~the whole realization is fully determined by the mapping of $\Gamma_{\omega}$ (i.e., of the first $\omega$ drawings of~the~realization).
\end{remark}

\subsection{Geometry and Topology of Graph Stories}\label{sse:geometry-topology}

The following lemma shows that the realizability problem is in essence more a topological problem than a geometric problem.

\begin{lemma2rep}\label{lem:topological-problem}
A graph story ${\cal S}=(G,\omega,k,\tau)$ is realizable on a set of points $P$, with $|P|=\omega +k$, if and only if it is realizable on any set of points $P'$ with~$|P'|=|P|$. 
\end{lemma2rep}
\begin{appendixproof}
Let ${\cal R} = \langle \Gamma_1, \Gamma_2, \dots, \Gamma_n \rangle$ be a realization of ${\cal S}$ on $P=\{p_1, p_2, \dots, p_{\omega+k}\}$. We show how to construct a realization ${\cal R}'$  of~${\cal S}$ on a given arbitrary set of points~$P'$ starting from ${\cal R}$. The realization ${\cal R}$ implicitly defines a function $\rho(\cdot)$ that for each edge $e$ of $G$ gives the Jordan arc $\rho(e)$ used to represent $e$. Let ${J}$ be the codomain of $\rho$, i.e., the set of Jordan arcs used by ${\cal R}$. We remark that ${J}$ is a set and not a multiset, in the sense that if two edges $e$ and $e'$ of $G$ are mapped to the same Jordan arc, then the curve $\rho(e)=\rho(e')$ is present in ${J}$ only once.
Without loss of generality, we may assume that any two Jordan arcs $c$ and $c'$ of ${J}$ have a finite intersection, i.e., either they do not cross, or they cross on a finite number of points. In fact, if two curves $c$ and $c'$ shared a portion of a curve, we could perturb one of them, say $c'$, so that $c'$ is drawn at an arbitrary small distance from $c$ until it crosses $c$ or diverges from it.

Starting from $P$ and ${J}$, we construct a multigraph ${\cal M}$: For each point $p_i \in P$, with $i=1,2,\dots, \omega+k$, ${\cal M}$ has a vertex $w_i$. For each Jordan arc $c \in {J}$ with endpoints $p_i$ and $p_j$, ${\cal M}$ has an edge between $w_i$ and $w_j$. Observe that the Jordan arcs in ${J}$ also provide a (non-planar) drawing $\Gamma({\cal M})$ of ${\cal M}$, where each pair of edges crosses a finite number of times.
We planarize ${\cal M}$ according to $\Gamma({\cal M})$ by replacing crossings with dummy vertices. Further, we subdivide multiple edges of ${\cal M}$ with degree-2 dummy vertices in order to obtain a plane graph ${\cal G}$.

We draw ${\cal G}$ by preserving its planar embedding on the set of points $P'$ plus an arbitrary set of additional points to host the planarization and subdivision dummy vertices; to this purpose we could use one of the algorithms described in  \cite{DBLP:journals/tcs/BadentGL08,DBLP:journals/gc/PachW01}. Let $\Gamma({\cal G})$ be the obtained planar drawing of ${\cal G}$.

Now observe that a vertex $v$ of $G$ corresponds to a point of $P$, which in turn is associated with a vertex of ${\cal M}$ and with a vertex of ${\cal G}$, which is drawn on a point of $P'$. 
Also, an edge $e$ of $G$ corresponds to a Jordan arc $\rho(e)$ in ${J}$, which in turn corresponds to an edge of ${\cal M}$ and to a simple path $\pi$ in ${\cal G}$ between the two points of $P'$ corresponding to the endpoints of $e$. Hence, we can define a function $\rho'(e)$ which gives, for each edge $e$ of $G$, a Jordan arc that is the concatenation of the curves used in $\Gamma({\cal G})$ to draw the path $\pi$.

Finally, observe that the Jordan arcs $\rho(e)$ and $\rho(e')$ of two edges $e$ and $e'$ of $G$ cross if and only if the corresponding paths in ${\cal G}$ share an intermediate vertex and, hence, if and only if $\rho'(e)$ and $\rho'(e')$ cross.

The other direction of the proof is obvious.
\end{appendixproof}
\begin{proofsketch}
Let ${\cal R}$ be a realization of ${\cal S}$ on $P$. Starting from ${\cal R}$, we construct a realization ${\cal R}'$ of~${\cal S}$ on a given arbitrary set of points~$P'$. Let $\rho(\cdot)$ be a function that for each edge $e$ of $G$ gives the Jordan arc $\rho(e)$ used by $\cal R$ to represent $e$ and let~$J$ be the codomain of $\rho$, i.e., the set (without repetitions) of Jordan arcs used by ${\cal R}$.
Without loss of generality, we may assume that any two Jordan arcs~$c$ and~$c'$ of ${J}$ have a finite intersection. This can be obtained by perturbing $c$ or~$c'$. 

Starting from $P$ and $J$, we construct a multigraph ${\cal M}$ that has a vertex $w_i$ for each point $p_i \in P$, with $i=1,2,\dots, \omega+k$, and an edge $(w_i,w_j)$ for each Jordan arc $c \in J$ with endpoints $p_i$ and $p_j$. Observe that the Jordan arcs in ${J}$ also provide a (non-planar) drawing $\Gamma({\cal M})$ of ${\cal M}$.
We planarize ${\cal M}$ by replacing crossings with dummy vertices. Further, we subdivide multiple edges of ${\cal M}$ in order to obtain a plane graph~${\cal G}$. By exploiting one of the algorithms described in~\cite{DBLP:journals/tcs/BadentGL08,DBLP:journals/gc/PachW01}, we can draw ${\cal G}$ while preserving its planar embedding on the set of points $P'$ plus an arbitrary set of additional points to host the planarization and subdivision dummy vertices, obtaining $\Gamma({\cal G})$.

Observe that a vertex $v$ of $G$ corresponds to a point of $P$, which is associated with a vertex of ${\cal G}$ drawn on a point of $P'$. 
Also, an edge $e$ of $G$ corresponds to a Jordan arc $\rho(e)$ in $J$, which is a simple path $\pi$ in ${\cal G}$. Hence, we define a function $\rho'(e)$ that gives, for each edge $e$ of $G$, a Jordan arc that is the concatenation of the curves used in $\Gamma({\cal G})$ to draw the path $\pi$. 
Finally, the Jordan arcs $\rho(e)$ and $\rho(e')$ of two edges $e$ and $e'$ of $G$ cross if and only if $\rho'(e)$ and $\rho'(e')$ cross.

The other direction of the proof is obvious.
\end{proofsketch}

It is natural to ask whether for every realizable graph story where $G$ is planar, there exists a planar embedding of~$G$ such that each drawing of the realization preserves this embedding. We formalize this concept and show that this is not always the case.
% From a topological perspective, it is natural to ask whether for every realizable graph story where~$G$ is planar there exists a planar embedding of~$G$ such that each drawing of the realization preserves this embedding. The next definition formalizes this concept and \cref{le:supporting-embedding} shows that this is not~always~the~case.
%
%\begin{defn}\label{de:supporting-embedding} 
Let ${\cal S}$ be a story whose graph $G$ is planar. A \emph{supporting embedding} for $\cal S$ is a planar embedding $\phi$ of $G$ such that $\cal S$ admits a realization $\langle \Gamma_1, \dots, \Gamma_n \rangle$ where the embedding of $\Gamma_i$ is the restriction of $\phi$ to $G_i$~($i=1, \dots, n$).
%\end{defn}

\begin{lemma2rep}\label{le:supporting-embedding}
There exists a minimal graph story ${\cal S}=(G,\omega,\tau)$ such that: $(i)$ $G$ is planar; $(ii)$ ${\cal S}$ is realizable; and $(iii)$ ${\cal S}$ does not admit a supporting embedding. 
\end{lemma2rep}
\begin{proofsketch}
We produce a minimal graph story ${\cal S}=(G,\omega,\tau)$ such that~$G$ admits a single planar embedding $\phi$ (up to a flip and up to the choice of the external face) and such that in any realization of ${\cal S}$ there is at least one embedding $\phi_i$ of $G_i$ that is not the restriction of $\phi$ to~$G_i$. In this story $\omega=8$, $G$ is the graph in  \cref{fi:no-supporting-embedding}, and $\tau$ is given by the indices of the vertices of $G$.
\begin{figure}[tb]
	\centering
	\subfigure[Graph $G$]{\label{fi:no-supporting-embedding-a}\includegraphics[page=1,width=0.35\textwidth]{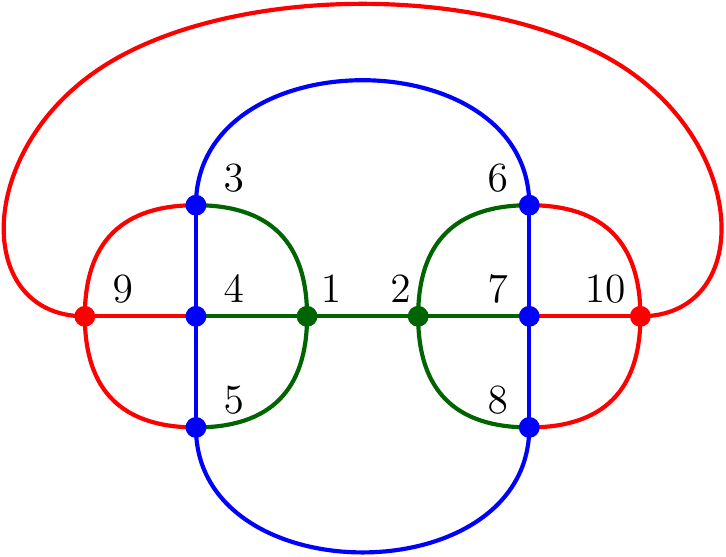}}
	\hfil
	\subfigure[Graph $G_8$]{\label{fi:no-supporting-embedding-b}\includegraphics[page=2,width=0.17\textwidth]{no-supporting-embedding}}
	\hfil
    \subfigure[Graph $G_9$]{\label{fi:no-supporting-embedding-c}\includegraphics[page=3,width=0.17\textwidth]{no-supporting-embedding}}
    \hfil
	\subfigure[Graph $G_{10}$]{\label{fi:no-supporting-embedding-d}\includegraphics[page=4,width=0.17\textwidth]{no-supporting-embedding}}
	\hfil
	\caption{An illustration for \cref{le:supporting-embedding}.}\label{fi:no-supporting-embedding}
\end{figure}
\end{proofsketch}
\begin{appendixproof}
We produce a minimal graph story ${\cal S}=(G,\omega,\tau)$ such that~$G$ admits a single planar embedding $\phi$ (up to a flip and up to the choice of the external face) and such that in any realization of ${\cal S}$ there is at least one embedding $\phi_i$ of $G_i$ that is not the restriction of $\phi$ to~$G_i$.
Consider the graph story ${\cal S}=(G,8,\tau)$ where $G$ is the graph depicted in \cref{fi:no-supporting-embedding-a} and $\tau$ is given by the indices of the vertices of $G$. Observe that $G$ is triconnected, and choose, without loss of generality, the embedding depicted in \cref{fi:no-supporting-embedding-a} for the realization of~$\mathcal{S}$. The restriction of such an embedding to $G_8$ provides the embedding $\phi_8$ of $G_8$. A drawing according to $\phi_8$ of $G_8$ is depicted in \cref{fi:no-supporting-embedding-b}. The drawings of $G_9$ and $G_{10}$ are obtained by deleting $v_1$ and $v_2$, respectively, and adding $v_9$ and $v_{10}$ in the positions where $v_1$ and $v_2$ are drawn in the drawing of $G_8$. \cref{fi:no-supporting-embedding-c,fi:no-supporting-embedding-d} show the obtained embeddings $\phi_9$ and $\phi_{10}$. It can be easily seen that both $\phi_9$ and $\phi_{10}$ are not the restriction of $\phi$ to $G_9$ and $G_{10}$, respectively.
\end{appendixproof}

\subsection{Characterizing Realizable Graph Stories}
We now give a characterization of a realizable graph story ${\cal S}=(G,\omega,k,\tau)$ in terms of a sequence of ``compatible embeddings''. To this aim, we give a generalization of the definition of planar embedding that associates with each face a weight representing how many of $k$ notable points are inside such a face.

%\textcolor{red}{inserire un raccordo e dire come si specializza questa definizione al caso di una storia minimale}
A \emph{face-$k$-weighted planar embedding $\phi$ of a planar graph $H$} is a planar embedding of $H$ together with a non-negative integer, called \emph{weight}, for each face of $\phi$ such that the sum of all weights is~$k$. 
%The concept of face-$k$-weighted planar embedding generalizes that of planar embedding, as a traditional planar embedding can be regarded as a face-$0$-weighted planar embedding.
%
%
The \emph{removal of a vertex $v$} from a face-$k$-weighted planar embedding $\phi$ of~$H$ produces a face-$(k+1)$-weighted planar embedding $\phi^{-v}$ of~$H \setminus v$ such that the planar embedding of $\phi^{-v}$ is the restriction of the planar embedding of $\phi$ to $H \setminus v$ and the weights of the faces are changed as follows: $(i)$ all the faces in common between $\phi$ and $\phi^{-v}$ have the same weight in $\phi^{-v}$ as in $\phi$, and $(ii)$ the new face of $\phi^{-v}$ resulting by the removal of $v$ has a weight that is one plus the sum of the weights of the faces of $\phi$~incident~to~$v$. 

Let ${\cal S}$ be a graph story and let $\phi_{i}$ be a face-$k$-weighted planar embedding of~$G_{i}$, for $i \in \{\omega, \dots, n\}$. 
Two face-$k$-weighted planar embeddings~$\phi_{i-1}$ and $\phi_i$, with $i=\omega+1, \dots, n$, are \emph{compatible} if removing $v_{i-\omega}$ from  $\phi_{i-1}$ produces the same face-$(k+1)$-weighted planar embedding of $G_{i-1} \cap G_i$ as removing $v_i$ from~$\phi_{i}$. 
%Observe that if $k = 0$ the above definition of compatible face-$k$-weighted embeddings coincides with the analogous definition of compatible embeddings for minimal graph stories. 

%Observe that the definition of face-$0$-weighted planar embedding (i.e., when ${\cal S}$ is minimal) coincides with the definition of planar embedding; also saying that two embeddings $\phi_{i-1}$ and $\phi_i$ are compatible equals to saying that their restrictions to $G_{i-1} \cap G_{i}$ have the same set of faces and the removal of $v_i$ from $\phi_i$ creates in the restriction to $G_{i-1} \cap G_{i}$ of $\phi_i$ the same face created by the removal of $v_{i-\omega}$ from $\phi_{i-1}$ in the restriction to $G_{i-1} \cap G_{i}$ of $\phi_{i-1}$.

\begin{lemma2rep}\label{le:compatibility-non-minimal}
A graph story ${\cal S}=(G,\omega,k,\tau)$ is realizable if and only if there exists a sequence $\langle \phi_{\omega}, \phi_{\omega+1}, \dots, \phi_{n} \rangle$ of face-$k$-weighted planar embeddings for the graphs $\langle G_{\omega}, G_{\omega+1}, \dots, G_{n} \rangle$, such that $\phi_{i-1}$ and $\phi_i$ are compatible $(\omega+1 \leq i \leq n)$.
\end{lemma2rep}
\begin{proofsketch}
We prove here only one direction.
Suppose there exists a sequence of face-$k$-weighted planar embeddings $\langle \phi_{\omega}, \phi_{\omega+1}, \dots, \phi_{n} \rangle$ such that any two consecutive face-$k$-weighted planar embeddings are compatible. 
Let $\Gamma_{\omega}$ be any planar drawing of $G_{\omega}$ and let $P$ be the set of points of $\Gamma_{\omega}$ corresponding to the vertices of $G_{\omega}$ plus $k$ unused points arbitrarily distributed inside the faces of $\Gamma_{\omega}$, according to the weights of $\phi_{\omega}$. 
For each $i = 1,  \dots, \omega-1$, define $\Gamma_i$ as the restriction of $\Gamma_{\omega}$ to~$G_i$. 
For each $i=\omega+1, \dots, n$, by the compatibility of $\phi_i$ with $\phi_{i-1}$, the removal of vertex $v_{i-\omega}$ from $\Gamma_i$ yields a drawing $\Gamma^{\cap}$ of $G^{\cap} = G_{i-1} \cap G_i$ that has the same face-$(k+1)$-weighted embedding $\phi^{\cap} = \phi_{i-1} \setminus v_{i-\omega} = \phi_i \setminus v_i$ of $G^{\cap}$.
We construct $\Gamma_i$ from $\Gamma^{\cap}$ by inserting $v_i$ inside the face of~$\Gamma^{\cap}$ corresponding to the face of~$\phi^{\cap}$ generated by the removal of $v_i$ from~$\phi_i$. Also, we 
planarly insert each edge connecting $v_i$ to each of its neighbors according to $\phi_i$, without changing the starting drawing and by leaving on each generated face $f$ the number of unused points that corresponds to the weight of $f$ in~$\phi_i$ (e.g., using the technique in~\cite{DBLP:journals/jgaa/ChanFGLMS15}, where the unused points are regarded as isolated vertices). The sequence $\langle \Gamma_1, \Gamma_2, \dots, \Gamma_n \rangle$ satisfies Properties~\textsf{R1} and~\textsf{R2}, i.e., it is a realization of $\cal S$.
\end{proofsketch}
\begin{appendixproof}
Suppose first that ${\cal S}$ is realizable and let ${\cal R} = \langle \Gamma_1, \dots, \Gamma_n \rangle$ be a realization of ${\cal S}$. Consider two consecutive drawings $\Gamma_{i-1}$ and $\Gamma_{i}$ of $G_{i-1}$ and $G_i$, respectively. Define a face-$k$-weighted planar embedding $\phi_{i-1}$ of $G_{i-1}$ ($\phi_{i}$ of $G_{i}$, respectively), where the planar embedding is that of $\Gamma_{i}$ ($\Gamma_{i+1}$, respectively) and the weight of each face is the number of unused points that the face contains in the drawing $\Gamma_{i}$ ($\Gamma_{i+1}$, respectively).
By Property~\textsf{R2} of \cref{df:realization}, the restrictions of $\Gamma_{i-1}$ and $\Gamma_{i}$ to $G_{i-1} \cap G_{i}$ is the same drawing, comprehensive of the positions of the $k+1$ unused points. We use the $k+1$ unused points of $G_{i-1} \cap G_{i}$ to define a face-$(k+1)$-weighted planar embedding $\phi^{\cap}$ of $G_{i-1} \cap G_{i}$, where the planar embedding is the one of $\Gamma_{i} \cap \Gamma_{i+1}$ and the weight of each face is the number of unused points it contains in the drawing $\Gamma_{i} \cap \Gamma_{i+1}$.
It is immediate to see that removing vertex $v_{i-\omega}$ from $\phi_{i-1}$ as well as removing vertex $v_i$ from $\phi_{i}$ produces in both cases $\phi^{\cap}$, i.e., $\phi_{i-1}$ and $\phi_{i}$ are compatible.

Suppose vice versa that there exists a sequence of face-$k$-weighted planar embeddings $\langle \phi_{\omega}, \phi_{\omega+1}, \dots, \phi_{n} \rangle$ such that any two consecutive face-$k$-weighted planar embeddings in the sequence are compatible. 
Let $\Gamma_{\omega}$ be any planar drawing of $G_{\omega}$ and let $P$ be the set of points of $\Gamma_{\omega}$ corresponding to the vertices of $G_{\omega}$ plus $k$ unused points arbitrarily distributed inside the faces of $\Gamma_{\omega}$ according to the weights of $\phi_{\omega}$. 
For each $i = 1,  \dots, \omega-1$, define $\Gamma_i$ as the restriction of $\Gamma_{\omega}$ to~$G_i$. 
For each $i=\omega+1, \dots, n$, by the compatibility of $\phi_i$ with $\phi_{i-1}$, we have that removing vertex $v_{i-\omega}$ from $\Gamma_i$ yields a drawing $\Gamma^{\cap}$ of $G^{\cap} = G_{i-1} \cap G_i$ that has the same face-$(k+1)$-weighted embedding $\phi^{\cap} = \phi_{i-1} \setminus v_{i-\omega} = \phi_i \setminus v_i$ of $G^{\cap}$.
We construct $\Gamma_i$ from $\Gamma^{\cap}$ by inserting $v_i$ inside the face of~$\Gamma^{\cap}$ corresponding to the face of~$\phi^{\cap}$ that is generated by the removal of $v_i$ from~$\phi_i$. Also, we 
planarly insert each edge connecting $v_i$ to each of its neighbors according to embedding~$\phi_i$, without changing the starting drawing and by leaving on each generated face $f$ the number of unused points that corresponds to the weight of $f$ in~$\phi_i$ (for example using the technique in~\cite{DBLP:journals/jgaa/ChanFGLMS15}, where the unused points are regarded as isolated vertices). The sequence $\langle \Gamma_1, \Gamma_2, \dots, \Gamma_n \rangle$ satisfies Properties~\textsf{R1} and~\textsf{R2}, i.e., it is a realization of $\cal S$ on $P$. 
\end{appendixproof}

\section{Realizability Testing of Graph Stories}\label{se:realizability-testing}

We first prove that testing whether a graph story is realizable is \NP-hard for any given integer $k \geq 0$ (\cref{th:hardness-non-minimal}). Then we prove the that the problem is in~\FPT~when parameterized by~$\sigma=\omega+k$ (\cref{th:fpt-non-minimal}).

\begin{theorem2rep}\label{th:hardness-non-minimal}
For any integer $k \geq 0$, testing the realizability of a graph story ${\cal S}=(G,\omega,k,\tau)$ is \NP-hard.
\end{theorem2rep}
\begin{proofsketch}
We use a reduction from the \textsc{Sunflower SEFE} problem, which is defined as follows. Let $G'_1, G'_2, \dots, G'_l$ be graphs on the same vertex-set such that each edge in the union of all graphs belongs either to only one of the input graphs or to all the input graphs.
\textsc{Sunflower SEFE} asks whether
%there exist $l$ planar drawings $\Gamma'_1,\Gamma'_2,\dots,\Gamma'_l$ of $G'_1, G'_2, \dots, G'_l$, respectively, such that: $(i)$ each vertex of $V'$ is mapped to the same point in every $\Gamma'_i$ ($1 \le i \le l$); $(ii)$ each edge belonging to all the input graphs is represented by the same simple curve in the drawings of all such graphs. In other words, the problem asks whether
there exists a drawing $\Gamma'$ of $G'_1 \cup G'_2 \cup \dots \cup G'_l$ such that two edges cross only if they do not belong to the same graph $G'_i$.
\textsc{Sunflower SEFE} is \NP-hard~for~$l \ge 3$~\cite{DBLP:journals/jgaa/Schaefer13}.

\begin{figure}[tb]
	\centering
	\subfigure[]{\label{fi:np-hardness2-a}\includegraphics[page=1,width=0.39\columnwidth]{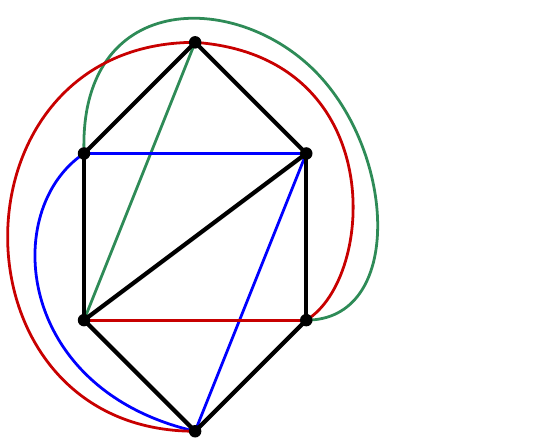}}
	\hfil
	\subfigure[]{\label{fi:np-hardness2-b}\includegraphics[page=2,width=0.39\columnwidth]{np-hardness2}}\\
	%\hfil
	\subfigure[$\Gamma_\omega$]{\label{fi:np-hardness2-c}\includegraphics[page=3,width=0.32\columnwidth]{np-hardness2}}
	\hfil
	\subfigure[$\Gamma_{\omega+2\tilde{\omega}}$]{\label{fi:np-hardness2-d}\includegraphics[page=4,width=0.32\columnwidth]{np-hardness2}}
	\hfil
	\subfigure[$\Gamma_{\omega+4\tilde{\omega}}$]{\label{fi:np-hardness2-e}\includegraphics[page=5,width=0.32\columnwidth]{np-hardness2}}
	\caption{(a) A drawing $\Gamma'$ of a \texttt{yes} instance $\{G'_1,G'_2,G'_3\}$ of \textsc{Sunflower SEFE}. The edges of $E'_\cap$ are black; the edges of $E'_1$, $E'_2$, and $E'_3$ are red, blue, and green, respectively. (b) A drawing of graph $G$ of the story ${\cal S}=(G,\omega,4,\tau)$ constructed from the instance of \cref{fi:np-hardness2-a}. The vertices of $D^A_1$ ($D^A_2$, $D^A_3$, resp.) and $D^B_1$ ($D^B_2$, $D^B_3$, resp.) are red (blue, green, resp.); the vertices of $\Delta_i$ ($1 \le i \le 4$) are purple.
	%(c) $\Gamma_\omega$. (d) $\Gamma_{\omega+2\tilde{\omega}}$. (e)~$\Gamma_{\omega+4\tilde{\omega}}$.
	The points of $P$ are represented as yellow~disks.} \label{fi:np-hardness2}
\end{figure}

Starting from an instance of \textsc{Sunflower SEFE} with $l=3$, we construct a non-minimal graph story 
${\cal S}=(G=(V,E),\omega,k,\tau)$ as follows; refer to \cref{fi:np-hardness2} for an example with $k=4$.
Let $G'_1,G'_2$, and $G'_3$ be the input graphs of \textsc{Sunflower SEFE} with vertex-set $V'$, let $E'_i$ be the set of edges that belong only to graph~$G'_i$ ($1 \le i \le 3$), and let $E'_\cap$ be the set of edges that belong to all the input graphs. Without loss of generality, we can assume that $|E'_1| \ge |E'_2| \ge |E'_3|$.

We now show how we define sets $V$ and $E$.
For every graph $G'_i$ ($1 \le i \le 3$), we subdivide each edge $e$ of $E'_i$ with two vertices $d_e^A$ and $d_e^B$ and we add them to two sets $D_i^A$ and $D_i^B$, respectively; see $d^A_{(u,y)}$ and $d^B_{(u,y)}$ in \cref{fi:np-hardness2-b}. We add the three edges obtained by subdividing $e$ to a set $E''_i$. If needed, we enrich sets $D_2^A$, $D_2^B$, $D_3^A$, and $D_3^B$ with isolated vertices so that all the sets $D_i^X$ have the same cardinality $\tilde{\omega} = |E'_1|$ (note that $\tilde{\omega} = |D_1^A| = |D_1^B|$), with $1 \le i \le 3$ and $X \in \{A,B\}$; see the green isolated vertices in \cref{fi:np-hardness2-b}. Also, we create four sets $\Delta_j$ ($1 \le j \le 4$) of $\tilde{\omega}$ isolated vertices $\delta_{j,1}, \dots, \delta_{j,\tilde{\omega}}$; see the purple isolated vertices in \cref{fi:np-hardness2-b}. We define $V = V' \cup D_1^A \cup D_1^B \cup D_2^A \cup D_2^B \cup D_3^A \cup D_3^B \cup \Delta_1 \cup \Delta_2 \cup \Delta_3 \cup \Delta_4$ and $E = E'_\cap \cup E''_1 \cup E''_2 \cup E''_3$.
We define $\omega$ as $\omega = |V'| + 6\tilde{\omega}$.
Finally, we suitably define $\tau$ in such a way that the vertices of the various subsets of $V$ appear in the following order: $\langle D_1^A, \Delta_1, D_2^A, \Delta_2, D_3^A, V', D_1^B, \Delta_3, D_2^B, \Delta_4, D_3^B \rangle$.

$\cal S$ is constructed in $O(|V'|)$ time and $\omega \in O(|V'|)$. We show in the appendix that $\cal S$ is realizable if and only if $\{G'_1,G'_2,G'_3\}$ is a~\texttt{yes} instance of \textsc{Sunflower SEFE}. Refer to \cref{fi:np-hardness2-c,fi:np-hardness2-d,fi:np-hardness2-e} for an example.
\end{proofsketch}
\begin{appendixproof}
We use a reduction from the \textsc{Sunflower SEFE} problem, which is defined as follows. Let $G'_1, G'_2, \dots, G'_l$ be graphs having the same vertex-set $V'$ such that each edge in the union of all graphs belongs either to only one of the input graphs or to all the input graphs.
The \textsc{Sunflower SEFE} problem asks whether
there exist $l$ planar drawings $\Gamma'_1,\Gamma'_2,\dots,\Gamma'_l$ of $G'_1, G'_2, \dots, G'_l$, respectively, such that: $(i)$ each vertex of $V'$ is mapped to the same point in every drawing $\Gamma'_i$ ($1 \le i \le l$); $(ii)$ each edge that is common to all the input graphs is represented by the same simple curve in the drawings of all such graphs. In other words, the problem asks whether there exists a drawing $\Gamma'$ of $G'_1 \cup G'_2 \cup \dots \cup G'_l$ such that two edges cross only if they do not belong to the same graph $G'_i$.
The \textsc{Sunflower SEFE} problem is \NP-complete for $l \ge 3$~\cite{DBLP:journals/jgaa/Schaefer13}.

% \begin{figure}[tb]
% 	\centering
% 	\subfigure[]{\label{fi:np-hardness2-a}\includegraphics[page=1,width=0.39\columnwidth]{np-hardness2}}
% 	\hfil
% 	\subfigure[]{\label{fi:np-hardness2-b}\includegraphics[page=2,width=0.39\columnwidth]{np-hardness2}}\\
% 	%\hfil
% 	\subfigure[$\Gamma_\omega$]{\label{fi:np-hardness2-c}\includegraphics[page=3,width=0.32\columnwidth]{np-hardness2}}
% 	\hfil
% 	\subfigure[$\Gamma_{\omega+2\tilde{\omega}}$]{\label{fi:np-hardness2-d}\includegraphics[page=4,width=0.32\columnwidth]{np-hardness2}}
% 	\hfil
% 	\subfigure[$\Gamma_{\omega+4\tilde{\omega}}$]{\label{fi:np-hardness2-e}\includegraphics[page=5,width=0.32\columnwidth]{np-hardness2}}
% 	\caption{(a) A drawing $\Gamma'$ of a \texttt{yes} instance $\{G'_1,G'_2,G'_3\}$ of \textsc{Sunflower SEFE}. The edges of $E'_\cap$ are black, while the edges of $E'_1$, $E'_2$, and $E'_3$ are red, blue, and green, respectively. (b) A drawing of graph $G$ of the graph story ${\cal S}=(G,\omega,4,\tau)$ constructed from the instance of \cref{fi:np-hardness2-a}. The vertices of $D^A_1$ ($D^A_2$, $D^A_3$, resp.) and $D^B_1$ ($D^B_2$, $D^B_3$, resp.) are red (blue, green, resp.); the vertices of $\Delta_i$ ($1 \le i \le 4$) are purple. (c) $\Gamma_\omega$. (d) $\Gamma_{\omega+2\tilde{\omega}}$. (e)~$\Gamma_{\omega+4\tilde{\omega}}$. The points of $P$ are represented as yellow disks.} \label{fi:np-hardness2}
% \end{figure}

Starting from an instance of the \textsc{Sunflower SEFE} problem with $l=3$, we construct a non-minimal graph story 
${\cal S}=(G=(V,E),\omega,k,\tau)$ as follows; refer to \cref{fi:np-hardness2} for an example with $k=4$.
Let $G'_1,G'_2$, and $G'_3$ be the input graphs of \textsc{Sunflower SEFE} having vertex-set $V'$, let $E'_i$ be the set of edges that belong only to graph $G'_i$ ($1 \le i \le 3$), and let $E'_\cap$ be the set of edges that belong to all the input graphs. Without loss of generality, we can assume that $|E'_1| \ge |E'_2| \ge |E'_3|$.

We now show how we define sets $V$ and $E$.
For every graph $G'_i$ ($1 \le i \le 3$), we subdivide each edge $e$ of $E'_i$ with two vertices $d_e^A$ and $d_e^B$ and we add them to two sets $D_i^A$ and $D_i^B$, respectively; see, e.g., $d^A_{(u,y)}$ and $d^B_{(u,y)}$ in \cref{fi:np-hardness2-b}. We add the three edges obtained by subdividing $e$ to a set $E''_i$. If needed, we enrich sets $D_2^A$, $D_2^B$, $D_3^A$, and $D_3^B$ with isolated vertices so that all the sets $D_i^X$ have the same cardinality $\tilde{\omega} = |E'_1|$ (note that $\tilde{\omega} = |D_1^A| = |D_1^B|$), with $1 \le i \le 3$ and $X \in \{A,B\}$; see, e.g., the green isolated vertices in \cref{fi:np-hardness2-b}. Also, we create four sets $\Delta_j$ ($1 \le j \le 4$) of $\tilde{\omega}$ isolated vertices $\delta_{j,1}, \dots, \delta_{j,\tilde{\omega}}$; see, e.g., the purple isolated vertices in \cref{fi:np-hardness2-b}. We define the set $V$ of vertices of $G$ as $V = V' \cup D_1^A \cup D_1^B \cup D_2^A \cup D_2^B \cup D_3^A \cup D_3^B \cup \Delta_1 \cup \Delta_2 \cup \Delta_3 \cup \Delta_4$ and the set $E$ of edges of $G$ as $E = E'_\cap \cup E''_1 \cup E''_2 \cup E''_3$.
We set the size $\omega$ of the viewing window as $\omega = |V'| + 6\tilde{\omega}$.
Finally, we suitably define $\tau$ in such a way that the vertices of the various subsets of $V$ appear in the following order: $\langle D_1^A, \Delta_1, D_2^A, \Delta_2, D_3^A, V', D_1^B, \Delta_3, D_2^B, \Delta_4, D_3^B \rangle$.

Observe that $\cal S$ can be constructed in $O(|V'|)$ time and that $\omega \in O(|V'|)$.~We now show that $\cal S$ is realizable if and only if the triplet $\{G'_1,G'_2,G'_3\}$ is a~\texttt{yes} instance of \textsc{Sunflower SEFE}. Refer to \cref{fi:np-hardness2-c,fi:np-hardness2-d,fi:np-hardness2-e} for an example.

\smallskip
\noindent($\Rightarrow$) Let $P$ be a set of $\omega+k$ points.
If $\cal S$ admits a realization on $P$, there exists a sequence of drawings $\langle \Gamma_1, \Gamma_2, \dots, \Gamma_{|V|} \rangle$ that satisfies Properties~\textsf{R1} and \textsf{R2}.

Drawing $\Gamma_\omega$ is induced by all vertices of $D_1^A$, $\Delta_1$, $D_2^A$, $\Delta_2$, $D_3^A$, $V'$, and $D_1^B$. This drawing is crossing-free by Property~\textsf{R1} and its restriction to the vertices of $V'$, $D_1^A$, and $D_1^B$ is a subdivision of a drawing of $G'_1$. Similarly, drawing $\Gamma_{\omega+2\tilde{\omega}}$ is induced by all vertices of $D_2^A$, $\Delta_2$, $D_3^A$, $V'$, $D_1^B$, $\Delta_3$, and $D_2^B$; its restriction to the vertices of $V'$, $D_2^A$, and $D_2^B$ is a subdivision of a (planar) drawing of $G'_2$. Finally, drawing $\Gamma_{\omega+4\tilde{\omega}} (= \Gamma_{|V|})$ is induced by all vertices of $D_3^A$, $V'$, $D_1^B$, $\Delta_3$, $D_2^B$, $\Delta_4$, and $D_3^B$; its restriction to the vertices of $V'$, $D_3^A$, and $D_3^B$ is a subdivision of a (planar) drawing of $G'_3$.

By Property~\textsf{R2}, $\Gamma_{\omega}$, $\Gamma_{\omega+2\tilde{\omega}}$, and $\Gamma_{\omega+4\tilde{\omega}}$ are such that their restrictions to their common subgraph are identical. Since $V'$ belongs to all these drawings, the edges in $E'_\cap$ are drawn identically in $\Gamma_{\omega}$, $\Gamma_{\omega+2\tilde{\omega}}$, and $\Gamma_{\omega+4\tilde{\omega}}$.

\smallskip
\noindent($\Leftarrow$) If the triplet $\{G'_1,G'_2,G'_3\}$ is a \texttt{yes} instance of \textsc{Sunflower SEFE}, there exists a drawing $\Gamma'$ of $G'_1 \cup G'_2 \cup G'_3$ such that two edges cross only if they do not belong to the same graph $G'_i$ ($1 \le i \le 3$).

The set~$P$ that we use for the realization of $\cal S$ is the union of the following sets: $(i)$ A set $P'$ consisting of the points of $\Gamma'$ to which the vertices of $V'$ are mapped, thus $|P'|=|V'|$; $(ii)$ a set $P_i^A$ ($P_i^B$, resp.) consisting of an arbitrarily chosen point $p_e^A$ ($p_e^B \neq p_e^A$, resp.) for each edge $e$ of $E'_i$ ($1 \le i \le 3$), such that $p_e^A$ ($p_e^B$, resp.) belongs to the simple curve representing $e$ in $\Gamma'$, thus $|P_i^A|=|E'_i|$ ($|P_i^B|=|E'_i|$, resp.); $(iii)$ a set $P_d$ consisting of $2(|E'_1|-|E'_2|) + 2(|E'_1|-|E'_3|)$ additional points;
$(iv)$ a set $P_k$ of $k$ additional points. Note that, $|P|=|V'|+6|E'_1|+k=|V'|+6\tilde{\omega}+k = \omega+k$.
We now prove that $\cal S$ admits a realization~on~$P$.

Recall that $V=V' \cup D_1^A \cup D_1^B \cup D_2^A \cup D_2^B \cup D_3^A \cup D_3^B \cup \Delta_1 \cup \Delta_2 \cup \Delta_3 \cup \Delta_4$ and that $\tau$ is such that the vertices appear in the following order: $\langle D_1^A, \Delta_1, D_2^A, \Delta_2, D_3^A, V', D_1^B, \Delta_3, D_2^B, \Delta_4, D_3^B \rangle$. Observe that, by construction, each non-isolated vertex of $D_i^X$ can be uniquely associated with a distinct point of $P_i^X$, where $1 \le i \le 3$ and $X \in \{A,B\}$. Namely, for each edge $e \in E'_i$, we associate vertex $d^A_e$ to $p^A_e$ and vertex $d^B_e$ to $p^B_e$.  
Drawing $\Gamma_\omega$ is such that: \begin{itemize}
	\item Each vertex of $D_1^A$ is mapped to the corresponding point of $P_1^A$;
	\item the vertices of $\Delta_1$ are distributed on all the points of $P_2^B$ and on $|E'_1|-|E'_2|$ points of $P_d$;
	\item each vertex of $D_2^A$ is mapped to the corresponding point of $P_2^A$;
	\item the vertices of $\Delta_2$ are distributed on all the points of $P_3^B$ and on $|E'_1|-|E'_3|$ points of $P_d$;
	\item each vertex of $D_3^A$ is mapped to the corresponding point of $P_3^A$;
	\item the vertices of $V'$ are mapped on the points of $P'$; 
	\item each vertex of $D_1^B$ is mapped to the corresponding point of $P_1^B$.
\end{itemize}

The edges of $\Gamma_\omega$ are the portions of the edges of $\Gamma'$ between points of $P$ to which the vertices of $\Gamma_\omega$ are mapped. $\Gamma_\omega$ is planar, since it is a subdivision of the planar subgraph of $\Gamma'$ induced by the edges of $E'_1 \cup E'_\cap$, plus some degree-1 vertices (the ones in $D_2^A$ and $D_3^A$), plus the isolated vertices in $\Delta_1$ and $\Delta_2$. Also, for $j=1,2, \dots, \omega-1$, define $\Gamma_j$ as the restriction of $\Gamma_\omega$ to $G_j$. This implies that the sequence $\langle \Gamma_1, \Gamma_2, \dots, \Gamma_\omega \rangle$ satisfies Properties~\textsf{R1} and~\textsf{R2}.

Drawing $\Gamma_{\omega+\tilde{\omega}}$ is obtained by replacing each vertex of $D_1^A$ with the isolated vertices of $\Delta_3$. The edges of $\Gamma_{\omega+\tilde{\omega}}$ are the portions of the edges of $\Gamma'$ between points of $P$ to which the vertices of $\Gamma_{\omega+\tilde{\omega}}$ are mapped. It is easy to see that Properties~\textsf{R1} and~\textsf{R2} are satisfied 
by the sequence $\langle \Gamma_\omega, \Gamma_{\omega+1}, \dots, \Gamma_{\omega+\tilde{\omega}} \rangle$.

In drawing $\Gamma_{\omega+2\tilde{\omega}}$, the vertices of $D_2^B$ replace the isolated vertices of $\Delta_1$, which were distributed on the points of $P_2^B$ and on $|E'_1|-|E'_2|$ points of $P_d$.
The edges of $\Gamma_{\omega+2\tilde{\omega}}$ are the portions of the edges of $\Gamma'$ between points of $P$ to which the vertices of $\Gamma_{\omega+2\tilde{\omega}}$ are mapped.
$\Gamma_{\omega+2\tilde{\omega}}$ is planar, since it is a subdivision of the planar subgraph of $\Gamma'$ induced by the edges of $E'_2 \cup E'_\cap$, plus some degree-1 vertices (the ones in $D_1^B$ and $D_3^A$), plus the isolated vertices in $\Delta_2$ and $\Delta_3$.
Also, for $j= \omega+\tilde{\omega}, \omega+\tilde{\omega}+1, \dots, \omega+2\tilde{\omega}-1$ define $\Gamma_j$ as the restriction of $\Gamma_{\omega+2\tilde{\omega}}$ to $G_j$. This implies that the sequence $\langle \Gamma_{\omega+\tilde{\omega}},\Gamma_{\omega+\tilde{\omega}+1}, \dots, \Gamma_{\omega+2\tilde{\omega}} \rangle$ satisfies Properties~\textsf{R1} and~\textsf{R2}.

Drawing $\Gamma_{\omega+3\tilde{\omega}}$ is obtained by replacing each vertex of $D_2^A$ with the isolated vertices of $\Delta_4$. The edges of $\Gamma_{\omega+3\tilde{\omega}}$ are the portions of the edges of $\Gamma'$ between points of $P$ to which the vertices of $\Gamma_{\omega+3\tilde{\omega}}$ are mapped. It is easy to see that Properties~\textsf{R1} and~\textsf{R2} are by the sequence $\langle \Gamma_{\omega+2\tilde{\omega}}, \Gamma_{\omega+2\tilde{\omega}+1},\dots,\Gamma_{\omega+3\tilde{\omega}} \rangle$.

In drawing $\Gamma_{\omega+4\tilde{\omega}} (= \Gamma_{|V|})$, the vertices of $D_3^B$ replace the isolated vertices of $\Delta_2$, which were distributed on the points of $P_3^B$ and on $|E'_1|-|E'_3|$ points of $P_d$.
The edges of $\Gamma_{\omega+4\tilde{\omega}}$ are the portions of the edges of $\Gamma'$ between points of $P$ to which the vertices of $\Gamma_{\omega+4\tilde{\omega}}$ are mapped.
$\Gamma_{\omega+4\tilde{\omega}}$ is planar, since it is a subdivision of the planar subgraph of $\Gamma'$ induced by the edges of $E'_3 \cup E'_\cap$, plus some degree-1 vertices (the ones in $D_1^B$ and $D_2^B$), plus the isolated vertices in $\Delta_3$ and $\Delta_4$.
Also, for $j=\omega+3\tilde{\omega},\omega+3\tilde{\omega}+1, \dots, \omega+4\tilde{\omega}-1$ define $\Gamma_j$ as the restriction of $\Gamma_{\omega+4\tilde{\omega}}$ to $G_j$. This implies that the sequence $\langle \Gamma_{\omega+3\tilde{\omega}}, \Gamma_{\omega+3\tilde{\omega}+1}, \dots, \Gamma_{\omega+4\tilde{\omega}} \rangle$ satisfies Properties~\textsf{R1} and~\textsf{R2}.

Observe that the obtained realization of $\cal{S}$ is such that, for each drawing $\Gamma_i$ ($i=1,\dots,|V|$) no vertex is mapped to a point of $P_k$.

\smallskip
We finally show that ${\cal S}$ admits a realization on a set of $\omega+k$ points if and only if it admits a realization on a set of $\omega$ points.
Clearly, if ${\cal S}$ admits a realization on $\omega$ points, it also admits a realization on a set of $\omega+k$ points.
For the other direction, suppose by contradiction that ${\cal S}$ admits a realization on a set of $\omega+k$ points and that it does not admit a realization on any set of $\omega$ points.
Note that in each drawing $\Gamma_i$ ($i=1,2,\dots,|V|$) of the realization, there are $k$ points to which no vertex is mapped.
Also, observe that $|D_1^A  \cup D_2^A \cup D_3^A \cup V' \cup D_1^B \cup D_2^B \cup D_3^B| = \omega$. Hence, there are at least $k$ points to which only vertices of $\Delta_1 \cup \Delta_2 \cup \Delta_3 \cup \Delta_4$ are mapped in some set of drawings of the realization. 
These vertices are isolated and thus they could have been mapped (without creating crossings) to points to which vertices of $D_1^A  \cup D_2^A \cup D_3^A \cup V' \cup D_1^B \cup D_2^B \cup D_3^B$ are mapped in some set of drawings of the realization. This implies that the $k$ extra-points could have not been used, and thus that ${\cal S}$ admits a realization on a set of $\omega$ points.
\end{appendixproof}

% \cref{th:fpt-non-minimal} proves the existence of an algorithm to test the realizability of a graph story. It implies that the problem is in \FPT~parameterized by~$\sigma=\omega+k$.

\begin{theorem2rep}\label{th:fpt-non-minimal}
Let ${\cal S}=(G,\omega,k,\tau)$ be a graph story and let $n$ be the number of vertices of $G$. There exists an 
$O(n \cdot 2^{(4 \sigma+1) \log_2 \sigma})$-time algorithm that tests whether~${\cal S}$ is realizable, where $\sigma=\omega+k$.
\end{theorem2rep}
\begin{proofsketch}
For each subgraph $G_i$ $(i = \omega, \dots, n)$, let ${\cal E}_i=\{\phi_i^1, \phi_i^2, \dots, \phi_i^{s_i}\}$~be the set of all planar face-$k$-weighted embeddings of $G_i$. We construct a directed acyclic graph $D$ as follows:
$(i)$ For each $\phi_i^j \in {\cal E}_i$ ($i = \omega, \dots, n$ and $j = 1, \dots, s_i$), $D$ has a node $v_i^j$ corresponding to $\phi_i^j$.   
$(ii)$ For each pair of elements $\phi_i^j$ and $\phi_{i+1}^r$ $(\omega \leq i \leq n-1; 1 \leq j \leq s_i; 1 \leq r \leq s_{i+1})$, $D$ contains~a~directed edge $(v_i^j,v_{i+1}^r)$ if and only if $\phi_i^j$ and $\phi_{i+1}^r$ are compatible face-$k$-weighted embeddings.

Each set ${\cal E}_i$, with $i=\omega, \dots, n$, defines a distinct \emph{layer} of vertices of~$D$, called \emph{layer $i$}.  
By construction, each vertex of layer $i$ can only have outgoing edges towards vertices of layer $i+1$ (if $i < n$) and incoming edges from vertices of layer $i-1$ (if $i > \omega$). We finally augment $D$ with a dummy source $s$ connected with outgoing edges to all vertices of layer 1 and with a dummy sink $t$ connected with incoming edges to all vertices of layer $n$. 
By \cref{le:compatibility-non-minimal}, $\cal S$ is realizable if and only if there is a directed path from $s$ to $t$ in $D$.
The time complexity of the algorithm is analyzed in the appendix.
\end{proofsketch}
\begin{appendixproof}
For each subgraph $G_i$ $(i = \omega, \dots, n)$, let ${\cal E}_i=\{\phi_i^1, \phi_i^2, \dots, \phi_i^{s_i}\}$~be the set of all planar face-$k$-weighted embeddings of $G_i$. We construct a DAG~(directed acyclic graph) $D$ as follows:
$(i)$ For each element $\phi_i^j \in {\cal E}_i$ ($i = \omega, \dots, n$ and $j = 1, \dots, s_i$), $D$ has a node $v_i^j$ corresponding to $\phi_i^j$.   
$(ii)$ For each~pair~of~elements $\phi_i^j$ and $\phi_{i+1}^r$ $(\omega \leq i \leq n-1; 1 \leq j \leq s_i; 1 \leq r \leq s_{i+1})$, $D$ contains~a~directed edge $(v_i^j,v_{i+1}^r)$ if and only if $\phi_i^j$ and $\phi_{i+1}^r$ are compatible face-$k$-weighted embeddings.

Each set ${\cal E}_i$, with $i=\omega, \dots, n$, defines a distinct \emph{layer} of vertices of~$D$, called \emph{layer $i$}.  
By construction, each vertex of layer $i$ can only have outgoing edges towards vertices of layer $i+1$ (if $i < n$) and incoming edges from vertices of layer $i-1$ (if $i > \omega$). We complete the construction of $D$ by adding a dummy source $s$ connected with outgoing edges to all vertices of layer 1  and a dummy sink $t$ connected with incoming edges to all vertices of layer $n$. 
By \cref{le:compatibility-non-minimal}, we have that $\cal S$ is realizable if and only if there is a directed path $\Pi$ from $s$ to $t$ in $D$. In fact $\Pi$ (if any) consists of exactly one vertex per layer, and the sequence of its vertices from layer $\omega$ to layer $n$ corresponds to a sequence of compatible face-$k$-weighted planar embeddings of $\langle G_\omega, G_{\omega+1}, \dots, G_n \rangle$.   

We now analyze the time complexity of the given algorithm.
The number $|{\cal E}_i|$ of nodes in layer $i$ equals the number of distinct face-$k$-weighted embeddings of~$G_i$. This number is the product of two factors: the number $\pi$ of possible planar embeddings of~$G_i$ times the number $\rho$ of ways you have to distribute $k$ units of weight among the faces of each planar embedding of~$G_i$. 
If~$G_i$ is connected, $\pi$ is upper bounded by the number of possible rotation systems for $G_i$, i.e., $\pi=O(\omega!)$. If $G_i$ is not connected, for each rotation system we also have to consider all possible ways of arranging a component inside the face of some other component. Since the number of faces is $O(\omega)$, this leads to $\pi = O(\omega! \cdot \omega^\omega)$ planar embeddings, which can be increased to $O(\omega! \cdot (\omega+k)^\omega)$. Observe that this space of planar embeddings is actually computable at the same cost 
%in $O(\omega! \cdot \omega^\omega)$ time, 
using SPQR-trees and BC-trees for describing the planar embeddings of each connected component and using an inclusion tree for describing the inclusion relationships among the different plane connected components. 
The number $\rho$ of ways you have to distribute $k$ units of weight among the faces of a planar embedding of~$G_i$ can be obtained (using the stars and bars metaphor popularized by \cite{feller1}) as $\rho = {{\omega-1+k}\choose{\omega-1}} = \frac{\omega}{\omega+k} {{\omega+k}\choose{\omega}}$. Since ${{n}\choose{m}} = O(\frac{n^m}{m!})$, we have $\rho = O(\frac{\omega}{\omega+k} \cdot \frac{(\omega+k)^\omega}{\omega!})$, which can be increased to $\rho = O(\frac{(\omega+k)^\omega}{\omega!})$.
Hence we have $\pi \cdot \rho = O(\omega! \cdot (\omega+k)^\omega \cdot \frac{(\omega+k)^\omega}{\omega!}) = O((\omega + k)^{2\omega}) = O( 2^{2\omega \log_2(\omega+k)}$).

Since we have $O(n)$ layers, generating the vertex set of $D$ takes $O(n \cdot \pi \cdot \rho)$. 
The number of edges of $D$ is $O(n \cdot (\pi \cdot \rho)^2)$, and checking whether we have to add an edge between two vertices of consecutive layers of~$D$ can be done in $O(\omega) = O(\omega+k)$ time, as we need to test the compatibility of the embeddings corresponding to the two vertices. Hence, generating the edge set of $D$ takes $O((\omega+k) \cdot n \cdot (\pi \cdot \rho)^2)$ time. Finally, checking whether a directed path from $s$ to $t$ exists in $D$ takes linear time in the size of~$D$. It follows that the whole testing algorithm takes 
$O((\omega+k) \cdot n \cdot (\pi \cdot \rho)^2) = O(n \cdot 2^{\log(\omega+k)} \cdot 2^{4(\omega+k) \log_2 (\omega + k)}) = O(n \cdot 2^{(4\sigma+1)\log_2 \sigma})$  time.
% formula di stirling https://en.wikipedia.org/wiki/Stirling%27s_approximation
\end{appendixproof}

\noindent When $k=0$, we have \cref{co:fpt-window-size}. Also, \cref{th:hardness-non-minimal,th:fpt-non-minimal} imply \cref{co:completeness-non-minimal}.

\begin{corollary}\label{co:fpt-window-size}
Let ${\cal S}=(G,\omega,\tau)$ be a minimal graph story and let $n$ be the number of vertices of $G$. There exists an  $O(n \cdot 2^{(4\omega+1) \log_2 \omega})$-time algorithm that tests whether~${\cal S}$ is realizable.
\end{corollary}

\begin{corollary}\label{co:completeness-non-minimal}
For any integer $k \geq 0$, testing the realizability of a graph story ${\cal S}=(G,\omega,k,\tau)$ is \NP-complete.
\end{corollary}

%%%%%%%%%%%%%%%%%%%%%%%%%%
% Minimal Graph Stories %%
%%%%%%%%%%%%%%%%%%%%%%%%%%

\section{Minimal Graph Stories}\label{se:minimal}
We now turn our attention to minimal graph stories that can be realized for small values of $\omega$. If $\omega \leq 4$ every minimal graph story is easily realizable, independent of~$G$ and of $\tau$, and even if $G$ is not a planar graph (just use any predefined planar drawing of the complete graph $K_4$ as a support for each $\Gamma_i$ $(i=1, \dots, n)$). Establishing which minimal graph stories are realizable when $\omega \geq 5$ is more challenging. We show that every graph story is realizable if $G$ is outerplanar (\cref{th:outerplanar-minimal}), while if $G$ is a series-parallel graph this is not always the case, even if $\omega=5$ (\cref{le:sp-minimal}). However, we prove that stories of partial 2-trees (which include series-parallel graphs) are always realizable for $\omega=5$ if we are allowed to ``reroute'' at most one edge per time (a formal definition is given~later); this result is an implication of a more general result for stories with $\omega=5$ (\cref{th:1-reroute-planar-minimal}). \cref{le:sp-minimal} and \cref{th:1-reroute-planar-minimal} together close the gap on the realizability of minimal graph stories of partial 2-trees when $\omega=5$.
Finally, for $\omega=5$ we prove that every minimal graph story is realizable if $G$ is a planar triconnected cubic graph (\cref{th:cubic3connected-minimal}). A graph is \emph{cubic} if all its vertices have degree three. 

\smallskip For a story of an outerplanar graph, we show that any outerplanar embedding is a supporting embedding (see the appendix for details).

\begin{theorem}\label{th:outerplanar-minimal}
Every minimal graph story ${\cal S}=(G,\omega,\tau)$ with $G$ outerplanar is realizable. Also, any outerplanar embedding of $G$ is a supporting embedding for~${\cal S}$.
%For an outerplanar graph $G$, every minimal graph story ${\cal S}=(G,\omega,\tau)$ is realizable, and any outerplanar embedding of $G$ is a supporting embedding for ${\cal S}$.
\end{theorem}
\begin{proof}
Let $\phi$ be any outerplanar embedding of $G$, and let $\phi_i$ be the restriction of $\phi$ to $G_i$ $(1 \leq i \leq n)$. Consider any two consecutive planar embeddings $\phi_{i-1}$ and $\phi_i$, for $\omega+1 \leq i \leq n$. Since they are restrictions of the same planar embedding of $G$, then their restrictions to $G_i \cap G_{i-1}$ determine the same set~$F$ of faces. Also, both $v_i$ and $v_{i-\omega}$ lie in the plane region corresponding to the external face of~$F$. Hence, $\phi_{i-1}$ and $\phi_i$ are compatible and, by \cref{le:compatibility-non-minimal}, ${\cal S}$ is realizable.
\end{proof}

\newlength{\tittoheight}
\setlength{\tittoheight}{0.165\columnwidth}

\begin{figure}[tb]
	\centering
	\subfigure[]{\label{fi:sp-unrealizable-story-1a}\includegraphics[page=1,height=\tittoheight]{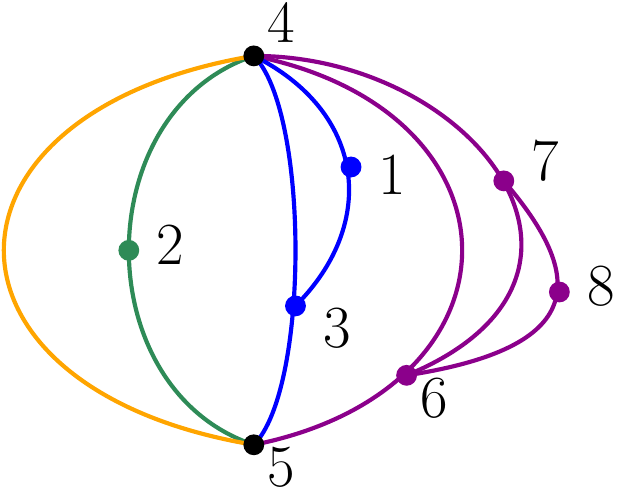}}
	\hfil
	\subfigure[]{\label{fi:sp-unrealizable-story-1b}\includegraphics[page=2,height=\tittoheight]{titto-proof}}
	\hfil
    \subfigure[]{\label{fi:sp-unrealizable-story-1c}\includegraphics[page=5,height=\tittoheight]{titto-proof}}
    \hfil
	\subfigure[]{\label{fi:sp-unrealizable-story-1d}\includegraphics[page=13,height=\tittoheight]{titto-proof}}
	    \hfil
	\subfigure[]{\label{fi:sp-unrealizable-story-1e}\includegraphics[page=9,height=\tittoheight]{titto-proof}}
	\hfil
	\caption{(a) A minimal graph story of a series-parallel graph that is not realizable. (b),(c),(d),(e) the four combinatorial
	embeddings of $G_5$.}\label{fi:sp-unrealizable-story-1}
\end{figure}

\setlength{\tittoheight}{0.205\columnwidth}
\begin{figure}[tbp]
	\centering
	\begin{tabular}{c c c c}
    \subfigure[$G_5$ of~\cref{fi:sp-unrealizable-story-1b}]{\label{fi:sp-unrealizable-story-b}\includegraphics[page=2,height=\tittoheight]{titto-proof}}
	&
	\subfigure[$G_6$]{\label{fi:sp-unrealizable-story-c}\includegraphics[page=3,height=\tittoheight]{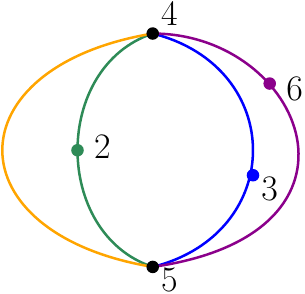}}
	&
	\subfigure[$G_7$ (failure)]{\label{fi:sp-unrealizable-story-d}\includegraphics[page=4,height=\tittoheight]{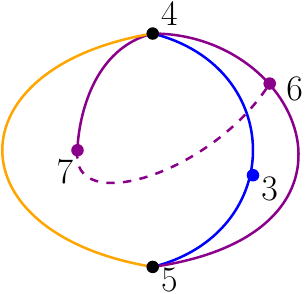}}
	\\
	\subfigure[$G_5$ of~\cref{fi:sp-unrealizable-story-1c}]{\label{fi:sp-unrealizable-story-e}\includegraphics[page=5,height=\tittoheight]{titto-proof}}
	&
	\subfigure[$G_6$]{\label{fi:sp-unrealizable-story-f}\includegraphics[page=6,height=\tittoheight]{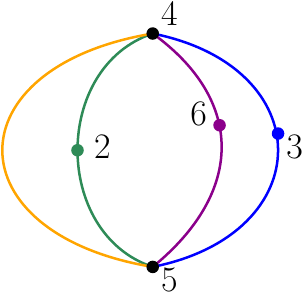}}
	&
	\subfigure[$G_7$]{\label{fi:sp-unrealizable-story-g}\includegraphics[page=7,height=\tittoheight]{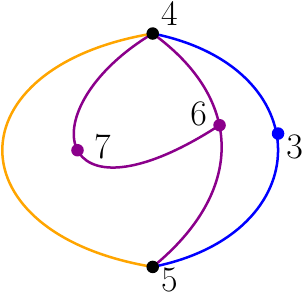}}
    &
	\subfigure[$G_8$ (failure)]{\label{fi:sp-unrealizable-story-h}\includegraphics[page=8,height=\tittoheight]{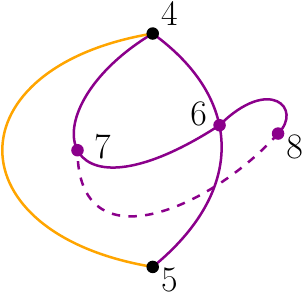}}
	\\
	\subfigure[$G_5$ of~\cref{fi:sp-unrealizable-story-1d}]{\label{fi:sp-unrealizable-story-i}\includegraphics[page=9,height=\tittoheight]{titto-proof}}
	&
	\subfigure[$G_6$]{\label{fi:sp-unrealizable-story-j}\includegraphics[page=10,height=\tittoheight]{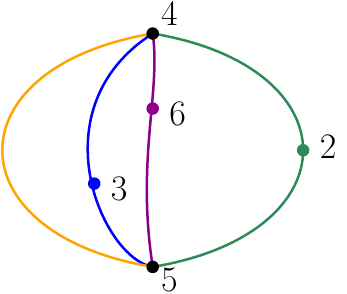}}
	&
	\subfigure[$G_7$]{\label{fi:sp-unrealizable-story-k}\includegraphics[page=11,height=\tittoheight]{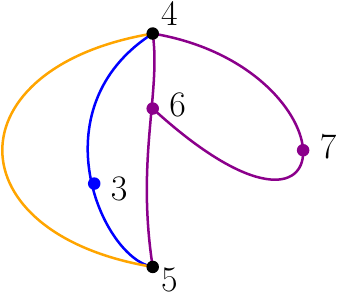}}
	&
	\subfigure[$G_8$ (failure)]{\label{fi:sp-unrealizable-story-l}\includegraphics[page=12,height=\tittoheight]{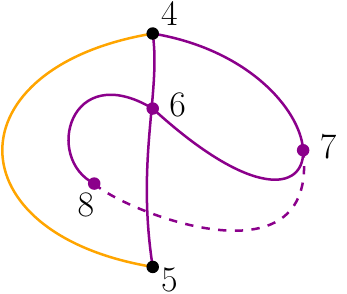}}
	\\
	\subfigure[$G_5$ of~\cref{fi:sp-unrealizable-story-1e}]{\label{fi:sp-unrealizable-story-m}\includegraphics[page=13,height=\tittoheight]{titto-proof}}
	&
	\subfigure[$G_6$]{\label{fi:sp-unrealizable-story-n}\includegraphics[page=14,height=\tittoheight]{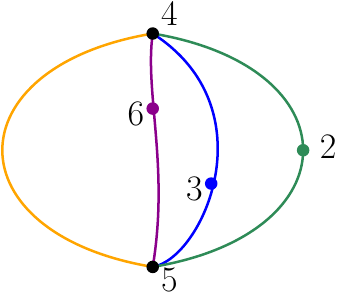}}
	&
	\subfigure[$G_7$ (failure)]{\label{fi:sp-unrealizable-story-o}\includegraphics[page=15,height=\tittoheight]{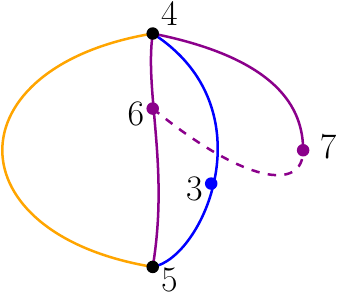}}
	\end{tabular}
	\caption{Tentative realizations of the story of \cref{fi:sp-unrealizable-story-1a} starting from the embeddings of \cref{fi:sp-unrealizable-story-1b,fi:sp-unrealizable-story-1c,fi:sp-unrealizable-story-1d,fi:sp-unrealizable-story-1e}. They all lead to a failure.}\label{fi:sp-unrealizable-story}
\end{figure}

\begin{lemma2rep}\label{le:sp-minimal}
For any $\omega \geq 5$, there exists a minimal graph story ${\cal S}=(G,\omega,\tau)$  such that $G$ is a series-parallel graph and ${\cal S}$ is not realizable. 
\end{lemma2rep}
\begin{proofsketch}
Consider the story $(G,5,\tau)$ in \cref{fi:sp-unrealizable-story-1a}, where the vertices are labeled with their subscript in the sequence $\tau = \langle v_1, v_2, \dots, v_8 \rangle$. 
Graph~$G_5$ admits one of the four embeddings in \cref{fi:sp-unrealizable-story-1b,fi:sp-unrealizable-story-1c,fi:sp-unrealizable-story-1d,fi:sp-unrealizable-story-1e}. Observe that, in all four cases either cycle $3,4,5$ separates $6$ from $7$ in $G_7$ (\cref{fi:sp-unrealizable-story-d,fi:sp-unrealizable-story-o}), or cycle $4,5,6$ separates $7$ from $8$ in $G_8$ (\cref{fi:sp-unrealizable-story-h,fi:sp-unrealizable-story-l}).
See the appendix for $\omega>5$.
\end{proofsketch}
\begin{appendixproof}
We first prove the statement for $\omega = 5$, and then extend the result to any $\omega > 5$.
Consider the instance ${\cal S}=(G,5,\tau)$ in \cref{fi:sp-unrealizable-story-1a}, where the vertices are labeled with their subscript in the sequence $\tau = \langle v_1, v_2, \dots, v_8 \rangle$. 
Graph~$G_5$ admits one of the four embeddings in \cref{fi:sp-unrealizable-story-1b,fi:sp-unrealizable-story-1c,fi:sp-unrealizable-story-1d,fi:sp-unrealizable-story-1e}. Observe that, in all four cases either cycle $3,4,5$ separates $6$ from $7$ in $G_7$ (\cref{fi:sp-unrealizable-story-d,fi:sp-unrealizable-story-o}), or cycle $4,5,6$ separates $7$ from $8$ in $G_8$ (\cref{fi:sp-unrealizable-story-h,fi:sp-unrealizable-story-l}).

To extend the result to any $\omega > 5$, we modify the above described instance. Consider the instance ${\cal S'}=(G',\omega',\tau')$, where $\omega' > 5$, $G'$ is obtained from $G$ by adding $\omega-5$ isolated vertices, and $\tau' = \langle v'_1, v'_2, \dots,$ $v'_{8+\omega-5} \rangle$ is such that for $p=1,2,\dots,5$, we have $v'_{p} = v_{p}$; for $q = 6,7,\dots, \omega$ we have that $v'_{q}$ is an isolated vertex of $G'$; and for $r = 1,2,3$ we have $v'_{\omega+1} = v_{5+r}$. 
Observe that $G'_5 = G_5$. From $G'_6$ to $G'_\omega$ the isolated vertices $v'_6, \dots v'_\omega$ are added to $G'_5$. Neglecting isolated vertices, for $r=1,2,3$, graph $G'_{\omega+r} = G_r$ and the same non-planarity configurations of the graph story ${\cal S}$ occur. 
\end{appendixproof}

\begin{figure}[tb]
	\centering
	\subfigure[]{\label{fi:rerouting-a}\includegraphics[page=1,width=0.18\columnwidth]{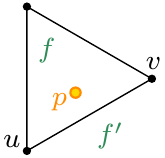}}
	\hfil
	\subfigure[]{\label{fi:rerouting-b}\includegraphics[page=2,width=0.18\columnwidth]{rerouting}}
	\caption{Rerouting edge $(u,v)$ with respect to point $p$.}\label{fi:rerouting}
\end{figure}

Since Property~\textsf{R2} of \cref{df:realization} is a strict requirement, one can think of relaxing it by allowing a partial change of the drawing of $G_{i-1} \cap G_i$ when vertex $v_i$ enters the viewing window. 
Let $\Gamma$ be a planar drawing of $G$, $(u,v)$ be an edge of $G$ incident to two distinct faces $f$ and $f'$ of $\Gamma$, and $p$ be a point of the plane inside face $f$; see \cref{fi:rerouting-a}. \emph{Rerouting $(u,v)$ with respect to $p$} consists of planarly redrawing $(u,v)$ such that $u$ and $v$ keep their positions and $p$ lies inside $f'$; see \cref{fi:rerouting-b}. The obtained drawing has the same planar embedding as $\Gamma$.

%\begin{defn}\label{df:h-reroute realization}
\noindent An \emph{h-reroute realization} of ${\cal S} = (G, \omega, k, \tau)$ \emph{on~a set} $P$ of $\omega+k$ points ($h \geq 0$) is a sequence $\langle \Gamma_1, \Gamma_2, \dots, \Gamma_n \rangle$ satisfying Property~\textsf{R1} of \cref{df:realization} and such that the restriction of $\Gamma_i$ to $G_{i-1} \cap G_i$ $(2 \leq i \leq n)$ is obtained from the restriction~of~$\Gamma_{i-1}$ to $G_{i-1} \cap G_i$ by rerouting at most $h$ distinct edges with respect to $h$ points of $P$.
\noindent ${\cal S}$ is \emph{h-reroute realizable} if it has an $h$-reroute realization on a set of $\omega+k$~points.
%A realization of~$\cal S$ is a $0$-reroute realization, i.e., $\cal S$ is realizable if and only if it is $0$-reroute realizable.
%end{defn}
%%%%%%%%%%%%%%%%%%%%%%%%%%%%%%%%%%%%%%%%%%%%%%%%%%%%

The next theorem characterizes the set of graph stories ${\cal S}=(G,5,\tau)$ that are 1-reroute realizable. It properly includes those stories whose $G$ is planar.

\begin{theorem2rep}\label{th:1-reroute-planar-minimal}
Every minimal graph story ${\cal S}=(G,5,\tau)$ is 1-reroute realizable if and only if $G$ does not contain $K_5$.
\end{theorem2rep}
\begin{proofsketch}
% Suppose that $G$ contains $K_5$. By assumption, all the edges of $G$ are visible, i.e., each edge appears in some $G_i$ $(1 \leq i \leq n)$. This implies that there must be an index $j \in \{1, \dots, n\}$ for which $G_j$ contains all vertices of the $K_5$, i.e., $G_j$ coincides with $K_5$. Indeed, if there were a pair of vertices of the $K_5$ that never appear in the same $G_i$, the edge connecting these two vertices would not be visible throughout any realization of the graph story. Therefore, since $G_j$ is non-planar, $G$ is not $h$-reroute realizable, for any $h \geq 0$. The other direction is more involved and it can be found in the appendix.
%
%versione nuova
%
We only sketch here the proof of one of the two directions. Suppose that $G$ does not contain $K_5$.
Let $\Gamma_4$ be a planar drawing of~$G_4$ on~$P$.
%Note that $\Gamma_4$ has at most six edges and at most four faces, which are pairwise adjacent.
Let $p$ be the point of $P$ to which no vertex of $G_4$ is mapped, let $f$ be the face of $\Gamma_4$ that contains $p$, and let $N(v_5)$ be the set of neighbors of $v_5$ in $G_5$.
%Note that $|N(v_5)| \le 4$, since the size of the viewing window is 5.
%Note that, if $G_4$ is $K_4$ then $|N(v_5)| \le 3$, as $G$ does not contain~$K_5$.
%
If the boundary of $f$ has four vertices, then $v_5$ can be mapped to $p$ and it can be connected to all its neighbors without creating edge crossings, so to obtain a planar drawing $\Gamma_5$ of~$G_5$.
If the boundary of $f$ has three vertices, mapping $v_5$ to $p$ and connecting it to its neighbors may create an edge crossing. To avoid this crossing, it is possible to reroute an edge of the boundary of $f$ with respect to $p$ such that $p$ lies inside a face whose boundary contains all vertices in $N(v_5)$. Such an edge always exists because the faces of $\Gamma_4$ are pairwise adjacent. More precisely, if $G_4$ is not $K_4$, then there is a face $f'$ of $\Gamma_4$ (adjacent to $f$) that contains all vertices of $G_4$.
%; in this case, we can reroute any edge $e$ shared by $f$ and $f'$ so that $p$ lies inside $f'$.
If $G_4$ is $K_4$, then $|N(v_5)| \le 3$, as $G$ does not contain~$K_5$. Also, there is a face $f'$ of $\Gamma_4$ that contains all vertices of $N(v_5)$. In both cases, we can reroute any edge $e$ shared by $f$ and $f'$ so that $p$ lies inside $f'$.
%After the rerouting operation, $v_5$ can be mapped to $p$ and connected to all its neighbors without creating edge crossings, so to obtain a planar drawing $\Gamma_5$ of~$G_5$.
%
This procedure can be applied for each pair of graphs $G_{i-1}$ and $G_i$ ($5 < i \le n$): $\Gamma_i$ is obtained by mapping $v_i$ to the same point $p$ of $P$ to which $v_{i-5}$ is mapped in $\Gamma_{i-1}$, by rerouting at most one edge with respect to $p$.
%The other direction can be found in the appendix.
\end{proofsketch}
\begin{appendixproof}
Suppose first that $G$ contains $K_5$. Recall that, by assumption, all the edges of $G$ are visible, i.e., each edge appears in some $G_i$ $(1 \leq i \leq n)$. This implies that there must be an index $j \in \{1, \dots, n\}$ for which $G_j$ contains all vertices of the $K_5$, i.e., $G_j$ coincides with $K_5$. Indeed, if there were a pair of vertices of the $K_5$ that never appear in the same $G_i$, the edge connecting these two vertices would not be visible throughout any realization of the graph story. Therefore, since $G_j$ is non-planar, $G$ is not $h$-reroute realizable, for any $h \geq 0$.

Suppose vice versa that $G$ does not contain $K_5$. This implies that each subgraph $G_i$ $(1 \leq i \leq n)$ does not contain $K_5$ and hence it is planar. To prove that $\cal S$ is 1-reroute realizable, we show that it admits a 1-reroute realization on any arbitrarily chosen set $P$ of 5 points. 
%Let $G_4$ be the subgraph of $G$ induced by all vertices $v_j$ such that $1 \le j \le 4$.
Let $\Gamma_4$ be a planar drawing of~$G_4$ on~$P$.
%Note that $\Gamma_4$ has at most six edges and at most four faces, which are pairwise adjacent.
Let $p$ be the point of $P$ to which no vertex of $G_4$ is mapped, let $f$ be the face of $\Gamma_4$ that contains $p$, and let $N(v_5)$ be the set of neighbors of $v_5$ in $G_5$.
%Note that $|N(v_5)| \le 4$, since the size of the viewing window is 5.
%Note that, if $G_4$ is $K_4$ then $|N(v_5)| \le 3$, as $G$ does not contain~$K_5$.
%
If the boundary of $f$ has four vertices, then $v_5$ can be mapped to $p$ and it can be connected to all its neighbors without creating edge crossings, so to obtain a planar drawing $\Gamma_5$ of~$G_5$.
If the boundary of $f$ has three vertices, mapping $v_5$ to $p$ and connecting it to its neighbors may create an edge crossing. To avoid this crossing, it is possible to reroute an edge of the boundary of $f$ with respect to $p$ such that $p$ lies inside a face whose boundary contains all vertices in $N(v_5)$. Such an edge always exists because the faces of $\Gamma_4$ are pairwise adjacent. More precisely, if $G_4$ is not $K_4$, then there must be a face $f'$ of $\Gamma_4$ (adjacent to $f$) such that $f'$ contains all vertices of $G_4$; in this case, we can reroute any edge $e$ shared by $f$ and $f'$ so that $p$ lies inside $f'$. If $G_4$ is $K_4$, then $|N(v_5)| \le 3$, as $G$ does not contain~$K_5$. Also, there is a face $f'$ of $\Gamma_4$ that contains all vertices of $N(v_5)$; as before, we can reroute any edge $e$ shared by $f$ and $f'$ so that $p$ lies inside $f'$.
After the rerouting operation, $v_5$ can be mapped to $p$ and connected to all its neighbors without creating edge crossings, so to obtain a planar drawing $\Gamma_5$ of~$G_5$.
This procedure can be applied for each pair of graphs $G_{i-1}$ and $G_i$ ($5 < i \le n$): $\Gamma_i$ is obtained by mapping $v_i$ to the same point $p$ of $P$ to which $v_{i-5}$ is mapped in $\Gamma_{i-1}$, by rerouting at most one edge with respect to $p$.
\end{appendixproof}

The proof of the next theorem is rather technical and can be found in the appendix. It relies on constructing a non-planar embedding $\phi$ of $G$ (with dummy vertices replacing crossings) such that there exists a realization $\langle \Gamma_1, \dots, \Gamma_n \rangle$ of $\cal S$ where the planar embedding of $\Gamma_i$ is the restriction of $\phi$ to $G_i$ ($i = 1, \dots, n$).

\begin{theorem2rep}\label{th:cubic3connected-minimal}
Every minimal graph story ${\cal S}=(G,5,\tau)$ such that $G$ is an $n$-vertex planar triconnected cubic graph is realizable. A sequence of compatible planar embeddings for ${\cal S}$ can be found in $O(n)$ time.
%in $O(n)$ time.
\end{theorem2rep}
\begin{appendixproof}
If $n\le 5$ the proof is trivial. Suppose $n \ge 6$.
We start by introducing some notation. Two vertices, a vertex and an edge, or two edges are \emph{coeval} if there is a graph $G_i$, for $1 \le i\le n$, containing both of them. By hypothesis every two adjacent vertices of $G$ are coeval. %\fabrizio{se questa ipotesi ha un nome, sarebbe bene citarla} We say that a cycle (a face, respectively) is \emph{coeval} with a vertex $v$ if all the edges belonging to the cycle (incident to the face, respectively) are coeval with $v$.

%In the following, we assume that an embedding of $G$ is described by the order of the edges around each vertex and by the set of crossings of the edges. A planar embedding has no crossing.

%\pino{definire G(C) e non-planar embedding. A-G: Fatto.}
In this proof we consider embeddings of $G$ in the plane. Since $G$ is triconnected, an embedding of $G$ (in the plane) is simply defined by the choice of the external face. In what follows, with the term ``embedding'' we also consider non-planar embeddings, where we interpret each crossing as a dummy vertex. When there is no crossing, we talk about planar embedding.

Given an embedding $\phi$ of $G$ and a cycle $C$ of $G$, we denote by $G(C)$ the union of the subgraph of $G$ that lies inside $C$ and $C$.
%i.e., it is not possible to connect a vertex of $G(C)$ to another vertex outside $G(C)$ without intersecting $C$.
%\fabrizio{forse vanno definiti meglio inside e outside di un ciclo}
%every line connecting a vertex of $G(C)$ to a point of the external face of $G$ intersects $C$.
A \emph{critical cycle} of $G$ in $\phi$ is a cycle $C$ such that there exists a vertex $v$ that is coeval with $C$ and $v \in G(C) \setminus C$. In this case we say that $C$ is critical for $v$. See, for example, cycles $C_1$ and $C_2$ in \cref{fi:cubic_graph-b}. A \emph{good embedding} of $G$ is an embedding with no critical cycle and where two coeval edges do not cross. See \cref{fi:cubic_graph-b,fi:cubic_graph-c}. 

%\pino{A me i claim sembrano sei (Giacomo: 2-5 sono tecnici, servono a provare 6. L'ho esplicitato)}
In order to prove the theorem, we use \cref{le:sufficient_cond,le:nocriticalinternal}: \cref{le:sufficient_cond} shows that a sufficient condition for story ${\cal S}=(G,5,\tau)$ to be realizable is that $G$ admits a good embedding; \cref{le:nocriticalinternal} shows that a planar triconnected cubic graph always admits a good embedding.

\begin{figure}[tb]
	\centering
	\hfil
	\subfigure[]{\label{fi:cubic_graph-a}}\includegraphics[width=0.3\columnwidth,page=1]{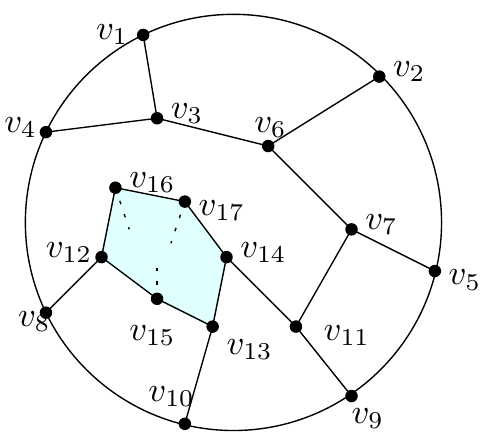}
	\hfil
	\subfigure[]{\label{fi:cubic_graph-b}}\includegraphics[width=0.3\columnwidth,page=2]{good_placement.pdf}
	\hfil
	\subfigure[]{\label{fi:cubic_graph-c}}\includegraphics[width=0.3\columnwidth,page=3]{good_placement.pdf}
	\caption{(a)~A planar triconnected cubic graph. Notice that it can be extended in the highlighted face. (b)~A planar embedding of $G$ that is not a good embedding, since it contains two critical cycles $C_1$ and $C_2$. (c)~A good embedding of $G$.}
	\label{fi:cubic_graph}
\end{figure}
\begin{lemma}
	\label{le:sufficient_cond}
	If $G$ admits a good embedding $\phi$, then ${\cal S}$ is realizable.
\end{lemma}
\begin{claimproof} Let $\phi_i$ be the restriction of $\phi$ to the subgraph $G_i$ ($1\le i\le n$). Consider $G$ embedded with $\phi$ and consider any $G_i$, for $1 \le i\le n$, embedded with $\phi_i$. For any $v_i$, let $G(v_i)=G_i\cup G_{i+1} \cup G_{i+2} \cup G_{i+3}\cup G_{i+4}$ and let $\phi(v_i)$ be the restriction of $\phi$ to the subgraph $G(v_i)$. Since every $G_i$ is embedded with $\phi_i$, $G(v_i)$ is embedded with $\phi(v_i)$. Notice that $G(v_i)$ may be not plane, while, since $\phi$ is a good embedding, every $G_j$, for $1\le j\le n$, is plane. Since $G$ has no critical cycle, we have that $v_i$ is incident to the external face of $G(v_i)$, which contains the external face of $G$. Hence, $v_i$ is incident to the external face of any $G_j$ for $i\le j\le i+4$. Since this holds for any $v_i$ with $1\le i \le n$, we have that any $G_j$, for $1\le j\le n$, is outerplanar. Hence, the proof follows by the same argument used in \cref{th:outerplanar-minimal}.
\end{claimproof}

Our goal now is to prove \cref{le:nocriticalinternal}. In order to do that we need intermediate results, stated by the following four claims. A critical cycle is \emph{internal} if it is not the external cycle (i.e., the cycle incident to the external face). \cref{clm:cubic_external_face} shows how to construct an embedding of $G$ having only internal critical cycles and where the crossings have specific properties.
\begin{clm}
%\luca{va specificato che la faccia esterna resta semplice?}
\label{clm:cubic_external_face}
There always exists an embedding $\phi'$ of $G$ that satisfies the following properties:~$(i)$ The external cycle
%\fabrizio{is simple and} 
is not a critical cycle;~$(ii)$ every two crossing edges are not coeval and one of them is an edge $e^*$ incident to $v_1$.
\end{clm}
\begin{claimproof}
% Since $\omega=5$, $v_1$ is coeval with four vertices, $\{v_2,v_3,v_4,v_5\}$. 
%Suppose that there exists a face $f$ of $G$ that is incident to $v_1, v_2, \dots, v_5$. We choose as $\phi$ the planar embedding of $G$ where $f$ is the external face. 
%
%Suppose now that there is no face that is incident to $v_1,$ $v_2, \dots, v_5$. 
%
First, observe that since $G$ is cubic triconnected and since every edge is incident to two coeval vertices, there is no face of $G$ incident to the vertices in the set $\{v_1,v_2,v_3,v_4,v_5\}$.

We first show that  there exists a face $f$ of $G$ incident to $v_1$ and to at least three vertices of the set $\{v_2,v_3,v_4,v_5\}$. Suppose that $v_1$ is not adjacent to $v_2$. Refer to \cref{fi:cubic-external-face-a}. In this case $v_1$ is adjacent to $v_3$, $v_4$, and $v_5$. Also, $v_2$ has to be adjacent to at least two vertices $v_i$ and $v_j$ among $v_3$, $v_4$, and $v_5$. Since in a triconnected cubic plane graph different from $K_4$ any 4-cycle is a face, there exists a face formed by the edges $(v_1,v_i)$, $(v_i,v_2)$, $(v_2,v_j)$, and $(v_j,v_1)$. Suppose now that $v_1$ is adjacent to $v_2$. Refer to \cref{fi:cubic-external-face-b}. There exist $i,i',i''\in \{3,4,5\}$ such that $v_1$ is adjacent to $v_i$ and $v_{i'}$; $v_{i''}$ is adjacent to $v_2$. Since $G$ is cubic: There exists a face $f$ incident to the path $\{v_{i''}, v_2,v_1\}$; either $(v_1,v_{i})$ or $(v_1,v_{i'})$ is incident to~$f$.

\begin{figure}[htb]
	\centering
	\subfigure[]{\label{fi:cubic-external-face-a}
	\includegraphics[width=0.25\columnwidth]{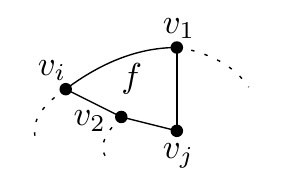}}
	\hfil
	\subfigure[]{\label{fi:cubic-external-face-b}
	\includegraphics[width=0.25\columnwidth,page=2]{external_not_critical.pdf}}
	\hfil
	\subfigure[]{\label{fi:cubic-external-face-c}
	\includegraphics[width=0.25\columnwidth,page=3]{external_not_critical.pdf}}
	\caption{Illustration for \cref{clm:cubic_external_face}.}
	\label{fi:cubic-external-face}
\end{figure}

Consider the planar embedding $\phi$ of $G$ when $f$ is the external face of $G$.  Let $a,b,c,d\in \{2,3,4,5\}$ and assume that $v_a$, $v_b$, and $v_c$ are adjacent to $v_1$. Refer to \cref{fi:cubic-external-face-c}.
%\titto{in 8(c) c'\`e un triangolo; in 8(b) no. Perch\'e?}
We have that two among these vertices, say $v_a$ and $v_b$, are in the external face. Hence, $v_c$ is not in the external face, otherwise $v_1, v_c$ would have been a separation pair. Vertex $v_d$ is in the external face due to the choice of $f$. If $(v_c,v_a)\in E$ and $(v_c,v_b)\in E$, $v_a,v_b$ is a separation pair. Hence, there exists a path $\pi$ connecting $v_c$ to either $v_a$ or $v_b$ containing only edges not coeval with $v_1$.
%\titto{non lo capisco}
Say that $\pi$ connects $v_c$ to $v_b$. Notice that since $(v_1,v_c)\in E$ and $e^*=(v_1,v_a)\in E$ is incident to $f$ and to an internal face incident to $v_c$. We can change $\phi$ so that $e^*$ crosses an edge in $\pi$ and $v_c$ is incident to the the external face. We call $\phi'$ the new embedding.

 Since $v_i$ for $i=1,\dots, 5$ is part of the external cycle, every internal vertex of $G$ is not coeval with $v_1$ and, consequently, Property~$(i)$ holds. Since we changed the embedding only with respect to $e^*$, Property~$(ii)$ holds for $\phi'$.
%
%Suppose, without loss of generality, that $v_1$ is incident to $v_2$, $v_3$, $v_5$. We have that two between these vertices are in the external face. Suppose, w.l.o.g., that $v_5$ is not in the external face. Since $v_4$ is in the external face, there exist a path  $\pi$ connecting either $v_5$ to $v_3$ or $v_5$ to $v_2$ containing only edges not coeval with any edge incident to $v_1$. Say that $\phi$ connects $v_5$ to $v_3$. We can change $\phi'$ so that $(1,2)$ crosses an edge in $\phi$ and $v_5$ is in the external face.  
\end{claimproof}

The next claims prove relevant properties of the internal critical cycles. 
\begin{clm}
	\label{cl:critical_cardinality}
	For every internal critical cycle $C$ of $G$, $6 \le |C| \le 8$.
\end{clm}
\begin{claimproof}
	Since $G$ is cubic triconnected and since $G(C)\not = C$, we have $|C| \ge 6$. Since every vertex is coeval with no more than 8 vertices, we have $|C| \le 8$.
\end{claimproof}

\begin{clm}
	\label{cl:critical_1containment}
	An internal critical cycle $C$ is critical for exactly one vertex.
\end{clm}
\begin{claimproof}
	Suppose that there are two vertices $v$ and $v'$ for which $C$ is critical. Since every vertex is coeval with 8 vertices and since $v$ and $v'$ are distinct vertices, in this case it is not possible that $|C|>6$. If $|C|=6$, there are two possibilities: (1)~$v$ and $v'$ are the only two vertices in $G(C)\setminus C$. In this case, there are exactly two vertices $x$ and $y$ of $C$ not adjacent to vertices in $G(C)\setminus C$. Vertices $x,y$ are a separating pair. (2)~Otherwise, in $C$ there are exactly three vertices incident to vertices of $G(C)\setminus C$. Let $v''$ be another vertex in $G(C)\setminus C$. We have that $v''$ cannot be coeval with both $v$ and $v'$ (the only vertices coeval with both $v$ and $v'$ are the ones of $C$). In both cases we have a contradiction.
\end{claimproof}

\begin{clm}
	\label{cl:nopasticca}
	Let $C$ be an internal critical cycle for vertex $v$. Given any other critical cycle $C'$ for a vertex $v'$, we have $(G(C')\setminus C')\cap (G(C)\setminus C)=\emptyset$.
\end{clm}
\begin{claimproof}
By \cref{cl:critical_1containment}, $C$ is not critical for $v'$ and $C'$ is not critical for $v$.
Suppose, by contradiction, that the statement does not hold.
We consider two different cases: (i)~$v'\not \in C$ and $v\not \in C'$; (ii)~$v'\in C$ or $v\in C'$.

\smallskip
\noindent \cal{Case (i).}  Suppose first that $C\not \in G(C')$ and $C' \not \in G(C)$; see, e.g., \cref{fi:cricial_disjoint-a}. In this case, since $G$ does not have vertices of degree 4, there are at least four vertices $x^a, x^b,x^c,x^d\in C \cap C'$ such that $x^i$ ($i\in \{a,b,c,d\}$) is incident to one edge in $C\cap C'$, one in $C'$ and not in $C$, and one in $C$ and not in $C'$. The external cycle of $C \cup C'$ has at least three vertices $y^a$, $y^b$, and $y^c$, adjacent to vertices outside it. Notice that, since $G$ is cubic, two vertices $x^i$ and $y^j$, with $i\in \{a,b,c,d\}$ and $j\in \{a,b,c\}$ cannot coincide. Also, in $C\cup C'$ there are three other vertices $z^a$, $z^b$, and $z^c$ connected to $v$ with three disjoint paths, due to the fact that $G$ is triconnected. Similarly, $C\cup C'$ contains three other vertices $z^d$, $z^e$, and $z^f$ connected to $v'$ with three disjoint paths.
Notice that it is possible to choose $z^i$ ($i\in \{a,\dots,f\}$) such that it does not coincide with any: $z^j$ ($j\neq i$ and $i\in \{a,\dots,f\}$); $x^k$ ($k\in \{a,\dots, d\}$); $y^g$ ($g \in \{a,b,c\}$). Hence, $|C\cap C'|\ge 4$ and $|C\cup C'|\ge 13$. It follows that either $|C|>8$ or $|C'|>8$, a contradiction. For example, in \cref{fi:cricial_disjoint-a} we have $|C'|= 9$.
%\giacomo{fino qui}

Suppose $C'\in G(C)$; see, e.g., \cref{fi:cricial_disjoint-b}. Vertex $v$ is coeval with: Its three adjacent vertices $x^a,x^b,x^c$; at least three vertices  $y^a,y^b,y^c$ in $C$ adjacent to three vertices outside $G(C)$; at least three vertices $z^a,z^b,z^c$  in $C$ adjacent to three vertices in $G(C')$. Notice that if there are less than three vertices of $C$ connected to vertices of $G(C')$, then the graph is not triconnected. In the figure, we have  $z^a,z^b,z^c\in C\cap C'$. If, for example, $C$ and $C'$ are disjoint, we have that $z^a,z^b,z^c$ are in $C$ (and not in $C'$). We have that $v$ is coeval with $9$ vertices. A contradiction. %, where $v\in G(C')$, and (c), where $v\not in G(C')$.
The case $C\in G(C')$ is analogous, where the role of $v,C$ and $v',C'$ is inverted.

%Suppose $C'\in G(C)$. The case $C\in G(C')$ is similar. Suppose $v\in G(C')$. Refer, for an example, to \cref{fi:cricial_disjoint-b}. In this case $v$ is coeval with: Its three adjacent vertices $x^a,x^b,x^c$ in $G(C')$; at least four vertices $y^a,y^b,y^c,y^d\in C$ such that $(y^a,y^b),(y^c,y^d)\in E$ and $y^a,y^c\in C'$; a vertex $z^a$ adjacent a vertex of $C'$ and a vertex $z^b$ adjacent to a vertex not in $G(C')$. Hence, $v$ is coeval with $9$ vertices. A contradiction. 

%Suppose $v\not \in G(C')$.  Refer to \cref{fi:cricial_disjoint-c}. We have that $C$ has the following vertices: Three vertices $x^a,x^b,x^c$ adjacent to a vertex outside $G(C)$; $y^a,y^b\in C'$; a vertex $z^a$ such that there is a path from $z^a$ to $v'$. Also, $v$ is adjacent to three vertices in $G(C)$ different from the above ones. Hence. $v$ is coeval with 9 vertices. A contradiction.

\begin{figure}[tb]
	\centering
	\subfigure[]{\label{fi:cricial_disjoint-a}}\includegraphics[width=0.4\columnwidth]{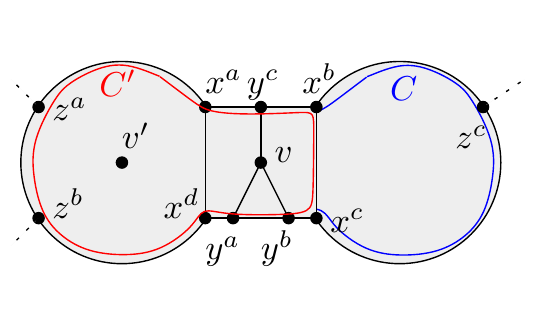}
	\hfil
	%\subfigure[]{\includegraphics[width=0.32\columnwidth,page=2]{critical_disjoint.pdf}}
	%\hfil
	\subfigure[]{\label{fi:cricial_disjoint-b}}\includegraphics[width=0.4\columnwidth,page=3]{critical_disjoint.pdf}
	\caption{Illustration for \cref{cl:nopasticca}, Case (i).}
	\label{fi:cricial_disjoint}
\end{figure}

\smallskip
\noindent \cal{Case (ii).} Refer to \cref{fi:nopasticca-a}. Observe that, since $G$ is planar cubic triconnected, the external cycle of $C\cup C'$ has the following distinct vertices:
\begin{itemize}
	\item (at least) three vertices $x^a$, $x^b$, and $x^c$ adjacent to a vertex outside the cycle, $y^a$, $y^b$, and $y^c$, respectively.
	\item six other vertices $z^a, z^b, z^c$ and $z^d,z^e,z^f$, such that $z^a, z^b, z^c$ ($z^d,z^e,z^f$) are connected with three disjoint path to $v$ ($v'$), respectively. %at least two between $z^a, z^b, z^c$, say $z^b$ and $z^c$, and two between $z^d,z^e,z^f$, say $z^d$ and $z^e$, are in $C\cap C'$.
\end{itemize}

Notice that one between $x^a$, $x^b$, $x^c$ has to be contained in $C$ (resp. $C'$) and not in $C'$ (resp. $C$), otherwise two between $z^a,z^b,z^c$ (resp. $z^d,z^e,z^f$) is a separation pair.  We consider two subcases: (iia)~$v\in C'$ and $v'\in C$; (iib)~$v \in C'$ and $v'\not \in C'$ (the case  $v \not\in C'$ and $v' \in C'$ is analogous).

%In this case at least two between $z^a, z^b, z^c$, say $z^b$ and $z^c$, and two between $z^d,z^e,z^f$, say $z^d$ and $z^e$, are in $C\cap C'$.Cycle $C$ ($C')$ contains $z^a$, $z^b$, and $z^c$ ($z^d$, $z^e$, and $z^f$). Notice that one between $x^a$, $x^b$, $x^c$ has to be contained in $C$ ($C'$) and not in $C'$ ($C$), otherwise $z^b,z^c$ ($z^e,z^d$) is a separation pair. Hence, either $|C'|=|C|=8$, or $|C'|=8$ and $|C|=7$, or $|C'|=7$ and $|C|=8$. We assume $|C'|=8$ and $|C|=7$, the other two cases can be solved similarly. We consider two subcases: (iia)~$v\in C'$ and $v'\in C$; (iib)~$v \in C'$ and $v'\not \in C'$ (the case  $v \not\in C'$ and $v' \in C'$ is analogous).
%
%Notice that if Cycle $C$ ($C')$ does not contain $z^a$, $z^b$, and $z^c$ ($z^d$, $z^e$, and $z^f$), $C$ contains two vertices outside $G(C\cup C')$ and $|C|>8$ ($|C'|>8$). This is not possible by  \cref{cl:critical_cardinality}. Hence, Cycle $C$ ($C')$ contains $z^a$, $z^b$, and $z^c$ ($z^d$, $z^e$, and $z^f$).  

\begin{figure}[tb]
	\centering
	\subfigure[]{\label{fi:nopasticca-a}}\includegraphics[width=0.32\columnwidth]{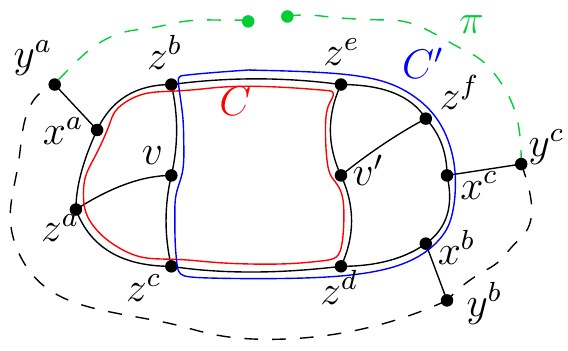}
	\hfil
	\subfigure[]{	\label{fi:nopasticca-b}}\includegraphics[width=0.32\columnwidth,page=9]{pasticca.pdf}
	\hfil
	\subfigure[]{	\label{fi:nopasticca-c}}\includegraphics[width=0.32\columnwidth,page=10]{pasticca.pdf}
	\caption{Illustration for \cref{cl:nopasticca}, Case (ii).}
	\label{fi:nopasticca}
\end{figure}

\medskip
\noindent
\emph{Case (ii.a):}
Suppose $v\in C'$ and $v'\in C$. Refer to \cref{fi:nopasticca-a}. 
In this case at least two between $z^a, z^b, z^c$, say $z^b$ and $z^c$, and two between $z^d,z^e,z^f$, say $z^d$ and $z^e$, are in $C\cap C'$. Cycle $C$ (resp. $C'$) contains $z^a$, $z^b$, and $z^c$ ($z^d$, $z^e$, and $z^f$).  Hence, either $|C'|=|C|=8$, or $|C'|=8$ and $|C|=7$, or $|C'|=7$ and $|C|=8$. We assume $|C'|=8$ and $|C|=7$, the other two cases can be solved similarly. Since $|C|=7$ and $|C'|=8$: $z^a\in C$ and $z^a\not \in C'$; $z^f\in C'$ and $z^f\not \in C$; $x^a\in C$ and $x^a\not \in C'$; $x^c,x^b\in C'$ and $x^c,x^b\not \in C$. We show that, in this case, $G$ contains an edge incident to two vertices that are not coeval. This would be a contradiction. Let $\tau(u)=i$ for any vertex $u=v_i$ of $G$. Suppose, w.l.o.g., $\tau(v)<\tau(v')$.

%In the rest of the proof, for any vertex $u=u_i$ of $G$, we denote by $\tau(u)=i$ .

We show that $\tau(y^a)<\tau(v)$. We have that $C'$ contains all the vertices coeval with $v'$.
Since $\tau(v)<\tau(v')$ and 
$v$ and $v'$ are coeval ($v\in C'$ and $v'\in C$), for any $u$ such that $\tau(u)>\tau(v)$ and $u$ is coeval with $v$, $u$ and $v'$ are coeval. Since $x^a$ is coeval with $v$ and not coeval with $v'$, we have $\tau(x^a)<\tau(v)$. If $\tau(y^a)>\tau(v)$, since $x^a$ and $v$ are coeval we have $y^a$ and $v$ coeval. In this case, $y^a$ would be coeval with $v'$, but this is not possible since all the vertices coeval with $v'$ are in $C'$. Hence,  $\tau(y^a)<\tau(v)$. 

We now show that $\tau(y^c)>\tau(v')$. Notice that, similarly to the previous case, since $\tau(v)<\tau(v')$ and since $v$ and $v'$ are coeval, for any $u$ such that $\tau(u)<\tau(v')$ and $u$ is coeval with $v'$, $u$ and $v$ are coeval. Let $w$ be the vertex coeval with $v$ and not in $C$ (recall that we assumed $|C|=7)$. We have that $w$ can be contained in at most one of the two paths connecting $y^a$ to $x^b$ and $x^c$ that are outside $C\cup C'$ and containing $y^b$ and $y^c$, respectively. Assume, w.l.o.g., that such path $\pi$ connecting $y^a$ to $x^c$ does not contain $w$. Hence, $w\not = x^c$ and, consequently, $x^c$ and $v$ are not coeval. It follows that $\tau(x^c)> \tau(v')$. If $\tau(y^c)<\tau(v')$, since $x^c$ and $v'$ are coeval, we have $y^c$ and $v'$ are coeval. In this case, $y^c$ would be coeval with $v$, but this is not possible since this implies $y^c=w$ but that contradicts the fact that $w\not\in \pi$. Hence, $\tau(y^c)>\tau(v')$.

We have $\tau(v)<\tau(v')$, $\tau(y^a)< \tau(v)$, $\tau(y^c)> \tau(v')$. Any vertex $u$ coeval with $v$ such that $\tau(u)>\tau(v)$ is in $G(C) \cup G(C')$ or it is $w$. Recall that $\pi$ connects $y^a$ to $y^c$ and does not contain $w$. Hence, $\pi$ contains an edge incident to two vertices that are not coeval. A contradiction.

\medskip
\noindent
\emph{Case (ii.b):} Suppose $v \in C'$ and $v'\not \in C'$. There are two cases: (1)~$G(C')\subset G(C)$; (2)~$G(C')\not \subset G(C)$, $G(C)\not \subset G(C')$, and $G(C')\cap G(C)\not = \emptyset$. 

(1) Refer to \cref{fi:nopasticca-b}. In this case we can choose $z^e$ and $z^d$ such that $C$ contains: $z^e$ or $z^d$, if $C'$ contains two between $x^c$ and $x^b$; both of them, as in the figure. Since $G(C')\subset G(C)$, $C$ contains  $x^b$, $x^c$, $y^b$, and $y^c$. Hence $|C|\ge 9$ and by \cref{cl:critical_cardinality} $C$ is not critical.

(2) Refer to \cref{fi:nopasticca-c}. There is a path (in the figure, an edge) such that $v$ and $v'$ are not in the same face. In this case the value of $|C'|$ increases by two with respect to its value in Case~(ii.a) and $|C'|\ge 9$ ($|C'|= 9$ if one between $x^b$ and $x^c$ is in $C$ and not in $C$, $|C'|= 10$ otherwise). Hence $|C'|\ge 9$ and by \cref{cl:critical_cardinality} $C'$ is not critical.

In both Cases~(i) and~(ii), we have a contradiction.
\end{claimproof}

Finally, we have all the ingredients to prove \cref{le:nocriticalinternal}.
\begin{lemma}
	\label{le:nocriticalinternal}
	Graph $G$ always admits a good embedding.
\end{lemma}
\begin{claimproof}
%\textcolor{red}{per la faccia esterna c'è un claim che prova che non è critica. usa l'altro claim di sopra per provare che due archi spostati con l'algoritmo di sotto non possono incrociare. vanno considerati i casi con $|C|=7$ e $|C|=8$ contenenti un $|C'|=6$} 
%
%\cref{clm:cubic_external_face,cl:critical_cardinality,cl:critical_1containment,cl:nopasticca}.
%
	Let $\phi$ be any embedding of $G$ with the properties described by \cref{clm:cubic_external_face}. We have that the external cycle of $G$ is not a critical cycle. In order to prove the claim, we consider all the critical cycles of $G$ and we modify $\phi$ in order to make them not critical, one by one. In \texttt{Part~I} of the proof we consider critical cycles that do not contain other critical cycles. In \texttt{Part~II}, we will show how to handle the rest of the cycles. During the first two parts of the proof we change $\phi$ with respect to some edges. In \texttt{Part~III} we show that these edges do not cross in the final embedding.
	
	\smallskip
	\texttt{Part I:}  Let $C$ be a critical cycle not containing other critical cycles. Let $v_k$ be the vertex of for which $C$ is critical. By \cref{cl:critical_1containment}, this vertex is unique.  Since any vertex $v_i$ such that $i\le 3$ or $i\ge n-2$ is coeval with less than $6$ vertices, by \cref{cl:critical_cardinality} we can assume $3<k< n-2$. 
	
	Suppose $|C|=6$. See \cref{fi:remove_critical_I-a}. Cycle $C$ always contains a vertex $v_i$ such that at least two vertices in $C$ are not coeval with $v_i$. In particular, if $C$ contains a vertex $v_j$ such that $j\in \{k-4,k+4\}$, $i=j$. Otherwise, $C=\{v_{k-3},v_{k-2},v_{k-1},$ $v_{k+1},v_{k+2},v_{k+3}\}$ and $i=k-3$ (or $i=k+3$). 	
	
	Consider the 6 faces adjacent to $C$, three internal and three external. %Let $F$ be the set of such faces. 
	Vertex $v_i$ is not coeval with some edge incident to 4 of the 6 vertices of $C$, hence, the subgraph of $G(C)$ consisting only of vertices coeval with $v_i$ consists of two faces and at least one of them, say $f$, is incident to all the vertices of $C$. See \cref{fi:remove_critical_I-b} Let $e=(v_i,v_{i'})$ be an edge of $C$ and assume that, going from $v_i$ to $v_{i'}$,  $v$ is on the right (left) of $e$. We can change the embedding around the vertices of $e$  so that $e$ goes through $f$ such that $v$ is on the left (right) of $e$. Cycle $C$ is not critical for~$v$ anymore and we did no introduce crossings between two coeval edges. See \cref{fi:remove_critical_I-c}.

\begin{figure}[htb]
	\centering
	\subfigure[]{\label{fi:remove_critical_I-a}}\includegraphics[width=0.28\columnwidth,page=2]{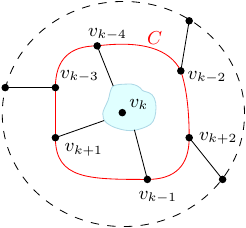}
	\hfil
	\subfigure[]{\label{fi:remove_critical_I-b}}\includegraphics[width=0.28\columnwidth,page=3]{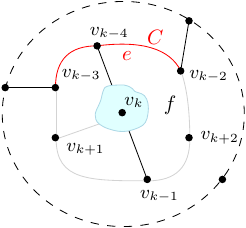}
	\hfil
	\subfigure[]{\label{fi:remove_critical_I-c}}\includegraphics[width=0.28\columnwidth,page=4]{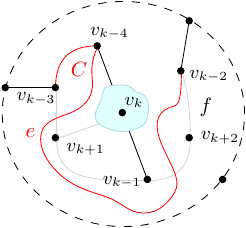}
	\caption{Illustration for \cref{cl:nopasticca}, \texttt{Part~I}.}
	\label{fi:remove_critical_I}
\end{figure}
	
	Suppose $|C|=7$ or $|C|=8$. In this case there is always a vertex $v_i$ that has 3 and 4 vertices of $C$ that are not coeval with $v$. In particular, we can always choose $v_i$ such that either $i=k-4$ or $i=k+4$. Hence, we always have that the remaining vertices coeval with $v_i$ are 4, as in the previous case, and we can prove these two new cases with the same considerations.
	
	%\cref{clm:cubic_external_face,cl:critical_cardinality,cl:critical_1containment,cl:nopasticca}.
	\smallskip
	\texttt{Part II:} Let $C^a$ and $C^b$ be two cycles such that $C^a\in G(C^b)$ critical for $v^a$ and $v^b$, respectively. By \cref{cl:nopasticca} and \cref{cl:critical_1containment}, $C^a$ and $C^b$ are critical for a same vertex $v=v^a=v^b$. Hence, by \cref{cl:critical_cardinality} and since $v$ is coeval with 8 vertices, $6\le C^a\cap C^b\le 7$. 
	
	Suppose, w.l.o.g., that $C^b\not \subset C^a$. We first show that $C^a\subset C^b$. Suppose, by contradiction, $C^a\not \subset C^b$. Since the graph is triconnected, in this case there is an edge connecting the vertices of $C^a$ not contained in $C^b$ to the vertices of $C^b$ not in $C^a$, otherwise there are two vertices $w$ and $w'$ that are a separation pair. See \cref{fi:remove_critical_II-a}. Refer now to \cref{fi:remove_critical_II-b}. Graph $C^a\cup C^b$ contains: Three vertices $x^a$, $x^b$, $x^v$ connected with a path to $v$; three vertices $y^a$, $y^b$, $y^c$ connected to the graph outside $G(C^a\cup C^b)$; (at least) four vertices $z^a$, $z^b$, $z^c$, $z^d$ incident to the three edges connecting $C^a$ to $C^b$, that have to be different from the previous vertices. Hence, $|C^a\cup C^b|\ge 10$ and $v$ is coeval with $10$ vertices. 
	
	We have $C^a\subset C^b$ and, consequently, $C^a$ is composed by vertices of $C^b$ and cords connecting two edges of $C^b$. See \cref{fi:remove_critical_II-c}. We can apply the approach of \texttt{Part~I}, with the difference that if $v_i$ is adjacent to one of these cords, we change the embedding with respect to the edge $e$ incident to $v_i$ in $C^a$ that is not the cord. In \cref{fi:remove_critical_II-c}, the edges in common between $C^a$ and $C^b$ are dashed. %Notice that $C^a\subset C^b$.
\begin{figure}[htb]
	\centering
	\subfigure[]{\label{fi:remove_critical_II-a}}\includegraphics[width=0.28\columnwidth,page=3]{critical_innested.pdf}
	\hfil
	\subfigure[]{\label{fi:remove_critical_II-b}}\includegraphics[width=0.28\columnwidth,page=2]{critical_innested.pdf}
	\hfil
	\subfigure[]{\label{fi:remove_critical_II-c}}\includegraphics[width=0.28\columnwidth,page=1]{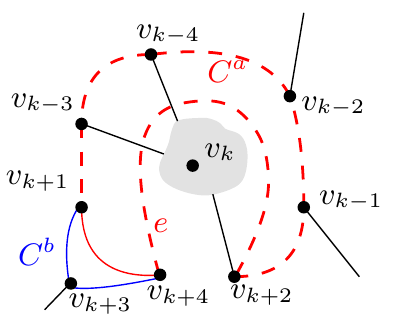}
	\caption{Illustration for \cref{cl:nopasticca}, \texttt{Part~II}.}
	\label{fi:remove_critical_II}
\end{figure}

	\smallskip
	\texttt{Part III:} Notice that in \texttt{Parts I} and \texttt{II} we changed the embedding with respect to exactly one edge $e$ (or simply, we ``moved one edge $e$'') for each critical cycle $C$. Also, after this operation, $e$ crosses only edges in $G(C)$. Hence,  By \cref{cl:nopasticca}, no two edges that we moved in \texttt{Parts I} and \texttt{II} cross. We have that, for any edge $e$ that we moved, $e$
	is never incident to $v_k$ such that $k\le 5$, since there exist no critical cycle of $G$ such that the vertex $v_i$ with maximum $i$ in $C$ is $k$.  Hence, by \cref{clm:cubic_external_face}, no edge moved in \texttt{Parts I} and \texttt{II} crosses the edge moved while processing the external face. %Hence, no two edges moved in \texttt{Parts I} and \texttt{II} cross. 
	\medskip
	   
	In \texttt{Parts I} and \texttt{II} we processed one by one each critical internal cycle in order to make it not critical, without creating crossings between coeval edges. We did it starting from an embedding where the external face was not critical. Hence, the obtained embedding, is a good embedding.
\end{claimproof}

The theorem holds by \cref{le:sufficient_cond,le:nocriticalinternal}. Concerning the computational time, choosing an embedding where the external face is not critical can be done in $O(1)$ time as described in the proof of \cref{clm:cubic_external_face}. Computing all the critical cycles can be done in $O(n)$ time. Fixing one by one each one of them, as described in the proof of \cref{le:nocriticalinternal}, takes linear time. This is sufficient by \cref{le:sufficient_cond}.
\end{appendixproof}

%\textcolor{red}{Can we say anything about lower bounds on $\omega$ for the realizability of a triconnected cubic.}

%\textcolor{red}{Maybe we can pose a question about universal embedding..}

%%%%%%%%%%%%%%%%%%%%%%%%%%%%%%
%% Non-minimal Graph Stories  %%
%%%%%%%%%%%%%%%%%%%%%%%%%%%%%%
\section{Lower Bounds for Non-minimal Graph Stories}\label{se:lower-bounds}

The next lemma can be used to prove lower bounds on the number of extra points required for the realizability of certain graph stories.

% \begin{lemma}\label{le:three-cycles}
% Let ${\cal S}=(G,\omega,k,\tau)$ be a realizable graph story.
% %
% Suppose that $G$ contains distinct cycles $c_1, \dots, c_h$ such that 
% %
% $c_{j-2}, c_{j-1}, c_j \in G_{i_j}$, with $j=3, \dots, h$ and $i_{j+1}-i_j<\omega$, 
% %
% and such that there is at least one vertex $v_j$ of $G$ with $v_j \in c_j$ and $v_j \not\in c_{j+1}$ and at least one vertex $v_{j+1}$ of $G$ with $v_{j+1} \not\in c_j$ and $v_{j+1} \in c_{j+1}$. 
% %
% Suppose also that in all planar embeddings of $G_{i_j}$, $c_j$ is outside $c_{j-1}$ that is outside $c_{j-2}$.
% %
% We have that $k\in \Omega(h)$.
% \end{lemma}

\begin{lemma2rep}\label{le:three-cycles}
Let ${\cal S}=(G,\omega,k,\tau)$ be a realizable graph story.
Suppose that: (i) $G$ contains vertex-disjoint cycles $C_1, \dots, C_h$ such that 
$C_{j-2}, C_{j-1}, C_j \in G_{i_j}$, with $j=3, \dots, h$, $i_{j-1} < i_{j}$ and $i_{j}-i_{j-1}<\omega$; (ii)
in all planar embeddings of $G_{i_j}$, $C_{j-1}$ separates $C_{j-2}$ from $C_j$.
We have that $\sigma=\omega+k \in \Omega(h)$.
\end{lemma2rep}
\begin{proofsketch}
Let ${\cal R} = \langle \Gamma_1, \dots, \Gamma_n \rangle$ be a realization of $\cal S$ and let $\sigma_i$ be the total number of points used by $\langle \Gamma_1, \dots, \Gamma_i \rangle$.
Without loss of generality, assume that in all planar embeddings of $G_{i_j}$, cycle $C_j$ is outside $C_{j-1}$, which is outside $C_{j-2}$. Also, observe that a cycle has at least $3$ vertices.
We prove, by induction on $j$, that the points used by $\cal R$ for the vertices in $\{C_1, \dots, C_{h-1}\}$ lie in the plane region delimited by~$C_h$ and that $\sigma_h \geq 9 + 3(h-1)$. See the appendix~for~details.\end{proofsketch}
\begin{appendixproof}
Let ${\cal R} = \langle \Gamma_1, \dots, \Gamma_n \rangle$ be a realization of $\cal S$ and let $\sigma_i$ be the total number of points used by $\cal R$ to draw $\langle \Gamma_1, \dots, \Gamma_i \rangle$.
Without loss of generality, assume that in all planar embeddings of $G_{i_j}$, cycle $C_j$ is outside $C_{j-1}$, which is outside $C_{j-2}$. Also, observe that a cycle has at least $3$ vertices.

The proof is by induction on $j$, showing that the points used by $\cal R$ to represent vertices in $\{C_1, \dots, C_{h-1}\}$ are contained in the region of the plane delimited~by~$C_h$ and that $\sigma_h \geq 9 + 3(h-1)$.
As a base case, consider any planar drawing $\Gamma_{i_3}$ of graph $G_{i_3}$. It contains vertex-disjoint cycles $C_1$, $C_2$, and $C_3$. 
Hence, by hypothesis the region of the plane delimited by $C_3$ contains the points used to represent $C_1$ and $C_2$.
Also, $\cal R$ uses at least $\sigma_{i_3} \geq 9$ points for the drawings of $\langle \Gamma_1, \dots, \Gamma_{i_3} \rangle$.

As for the inductive case, consider graph $G_{i_j}$ and any of its drawings $\Gamma_{i_j}$. By the inductive hypothesis, the points used by $\cal R$ to represent vertices in $\{C_1, \dots, C_{{i_{j-1}}-1}\}$ are contained in the region of the plane delimited by $C_{i_{j-1}}$.
Also, the realization $\cal R$ uses $\sigma_{i_{j-1}} \geq 9 + 3(({j-1})-1)$ points for drawings $\langle \Gamma_1, \dots, \Gamma_{i_{j-1}} \rangle$.  
Drawing $\Gamma_{i_j}$ contains cycles $C_{i_{j-2}}$, $C_{i_{j-1}}$, and $C_{i_j}$. By hypothesis, we have that cycle $C_{i_j}$ is outside $C_{i_{j-1}}$, that is outside $C_{i_{j-2}}$ in $\Gamma_{i_j}$. There exist at least three vertices that belong to $C_{i_j}$ and are outside $C_{i_{j-1}}$. Therefore, the points used by~$\cal R$ to represent vertices in $\{C_1, \dots, C_{i_j}\}$ are contained in the region of the plane delimited by $C_{i_j}$ and that $\sigma_{i_j} \geq \sigma_{i_{j-1}} + 3 \geq  9 + 3(j-1)$.
\end{appendixproof}

%The following lemma allows us to easily prove upper bounds for $k$ for certain graph stories.
%
%\begin{lemma}
%Let ${\cal S}=(G,\omega,k,\tau)$ be a %realizable graph story and suppose that $G$ has a planar embedding with at most $h$ vertices on the same face.
%%
%There exists a $k' \leq |V|-h $ such that ${\cal S}' = (G,\omega,k',\tau)$ is realizable.
%%We have that $k \in O(|V|-h)$.
%\end{lemma}

As an example, we exploit \cref{le:three-cycles} to prove the following theorem, which generalizes~\cite[Theorem 1]{bddfp-gsisa-jgaa20} and whose proof can be found in the appendix.

Let $n \equiv 0 \mod 3$. An \emph{$n$-vertex nested triangles~graph} $G$ contains the vertices and edges of the 3-cycle $C_i = (v_{i-2},v_{i-1},v_i)$, for $i = 3,6,\dots,n$, plus the edges $(v_i,v_{i+3})$, for $i = 1,2,\dots,n-3$. For $n \geq 6$, $G$ is triconnected, thus it has a unique planar embedding (up to the choice of the external face) \cite{Whitney33}.

\begin{theorem2rep}\label{th:nested-triangulation}
Let ${\cal S} =(G,9,k,\tau)$ be a realizable graph story such that $G$ is a $3h$-vertex nested triangles graph, where $\tau$ is given by the indices of the vertices of~$G$. Any realization of $\cal S$ has $k \in \Omega(n)$, where $n=3h$ is the number of vertices~of~$G$.
\end{theorem2rep}
\begin{appendixproof}
Consider the vertex-disjoint cycles $C_1 = (v_1, v_2, v_3)$, $C_2 = (v_4, v_5, v_6)$, \dots, $C_i = (v_{3i-2}, v_{3i-1}, v_{3i})$, \dots, $C_h = (v_{3h-2}, v_{3h-1}, v_{3h})$. 
We have that $G_9$ contains cycles $C_1, C_2$, and $C_3$. Also, in any planar embedding of $G_9$ we have that $C_2$ separates $C_1$ from $C_3$. More generally, for $j= 3, 4, \dots, h$ graph $G_{3j}$ contains cycles $C_{j-2}$, $C_{j-1}$, and $C_{j}$, and in any planar embedding of $G_{3j}$ we have that $C_{j-1}$ separates $C_{j-2}$ from $C_j$.
By \cref{le:three-cycles}, we have that $\omega+k \in \Omega(h) \in \Omega(n)$. Since $\omega$ is a constant, we have that $k \in \Omega(n)$.
\end{appendixproof}

While \cref{le:three-cycles} exploits the uniqueness of the embedding of $G$, the next result provides lower bounds also for graphs that have several~planar~embeddings.

\begin{figure}[tb]
    \centering
    \includegraphics[page=1, width=0.3\columnwidth]{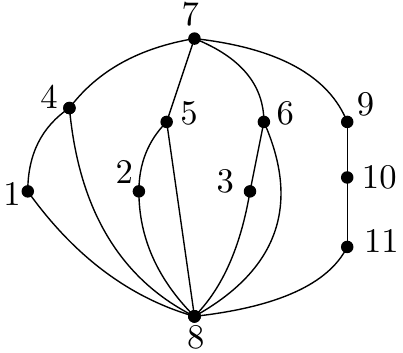}
    \caption{Illustration for \cref{th:sp-non-minimal}. Case $\omega=8$.}
    \label{fi:non-minimal-8-G}
\end{figure}

\begin{theorem2rep}\label{th:sp-non-minimal}
For any $\omega \geq 8$, there exists a graph story ${\cal S}=(G,\omega, k,\tau)$ such that $G$ is a series-parallel graph and ${\cal S}$ is not realizable for $k < \lfloor\frac{\omega}{2}\rfloor-3$. 
\end{theorem2rep}
\begin{proofsketch}
We prove here the statement for $\omega = 8$, and we show in the appendix how to extend the result to any $\omega > 8$. Consider the instance ${\cal S}=(G,8,0,\tau)$ in \cref{fi:non-minimal-8-G}, where the vertices are labeled with their subscript in the order $\tau = \langle v_1, v_2, \dots, v_{11} \rangle$. Graph~$G$ is a parallel composition of four components, three of which are a series of an edge and a triangle, and the other one is a path of length four.
\begin{figure}[tb]
    \centering
     \subfigure[$G_8$ for $k=0$]{\label{fi:non-minimal-story-8-failure-G8}\includegraphics[page=2,width=0.24\columnwidth]{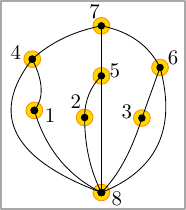}}
     \hfill
     \subfigure[$G_{11}$ for $k=0$ (failure)]{\label{fi:non-minimal-story-8-failure-G11}\includegraphics[page=3,width=0.24\columnwidth]{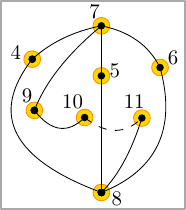}}
     \hfill
     \subfigure[$G_8$ for $k=1$]{\label{fi:non-minimal-story-8-solution-G8}\includegraphics[page=4,width=0.24\columnwidth]{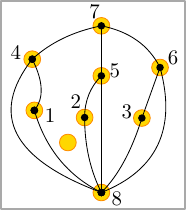}}
     \hfill
     \subfigure[$G_{11}$ for $k=1$]{\label{fi:non-minimal-story-8-solution-G11}\includegraphics[page=5,width=0.24\columnwidth]{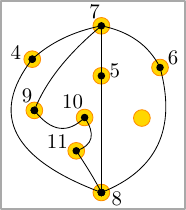}}
    \caption{Illustration for \cref{th:sp-non-minimal}. Case $\omega=8$. Drawings of $G_8$ and $G_{11}$.}
    \label{fi:non-minimal-8}
\end{figure}

Observe that, in any planar embedding of $G_\omega=G_8$ at most two among $v_1, v_2$, and $v_3$ can be incident to the same face (see \cref{fi:non-minimal-story-8-failure-G8}).
Graph $G_{11}$ contains the paths $(v_7, v_4, v_8), (v_7, v_5, v_8), (v_7, v_6, v_8)$, and $(v_7, v_9, v_{10}, v_{11}, v_8)$. Since $v_9, v_{10}$, and $v_{11}$ are mapped to the points where $v_1, v_{2}$ and $v_3$ are mapped, respectively, it is not possible to obtain a planar embedding of $G_{11}$
%with such a constraint
(see \cref{fi:non-minimal-story-8-failure-G11}). 
%It follows that either in any embedding of $G_{10}$ there is a cycle separating $v_9$ from $v_{10}$ or in any embedding of $G_{11}$ there is a cycle separating $v_{1}$ from $v_{11}$.
Thus, $\cal S$ does not admit a realization. 

To prove that ${\cal S}=(G,8, k,\tau)$ is realizable for $k \geq \lfloor\frac{\omega}{2}\rfloor-3 = 1$, suppose~that~$v_1$ and $v_2$ are drawn on the same face $f$ and there is an extra point $p$ inside~$f$. In~this case $\cal S$ is realizable, and $G_8$ and $G_{11}$ are drawn as in~\cref{fi:non-minimal-story-8-solution-G8,fi:non-minimal-story-8-solution-G11}.
\end{proofsketch}
\begin{appendixproof}
We first prove the statement for $\omega = 8$, and then we extend the result to any $\omega > 8$. Consider the instance ${\cal S}=(G,8,0,\tau)$ in \cref{fi:non-minimal-8-G}, where the vertices are labeled with their subscript in the order $\tau = \langle v_1, v_2, \dots, v_{11} \rangle$. Graph~$G$ is a parallel composition of four components, three of which are a series of an edge and a triangle, which we call \emph{flags}, and the other one is a path of length four.
Observe that, in any planar embedding of $G_\omega=G_8$ at most two among $v_1, v_2$, and $v_3$ can be incident to the same face (see, e.g., \cref{fi:non-minimal-story-8-failure-G8}).
Graph $G_{11}$ contains the paths $(v_7, v_4, v_8), (v_7, v_5, v_8), (v_7, v_6, v_8)$, and $(v_7, v_9, v_{10}, v_{11}, v_8)$. Since $v_9, v_{10}$, and $v_{11}$ are mapped to the points where $v_1, v_{2}$ and $v_3$ are mapped, respectively, it is not possible to obtain a planar embedding of $G_{11}$
%with such a constraint
(see, e.g., \cref{fi:non-minimal-story-8-failure-G11}). 
%It follows that either in any embedding of $G_{10}$ there is a cycle separating $v_9$ from $v_{10}$ or in any embedding of $G_{11}$ there is a cycle separating $v_{1}$ from $v_{11}$.
Thus, $\cal S$ does not admit a realization. 

To prove that ${\cal S}=(G,8, k,\tau)$ is realizable for $k \geq \lfloor\frac{\omega}{2}\rfloor-3 = 1$, suppose~that $v_1$ and $v_2$ are drawn on the same face $f$ and there is an extra point $p$ inside~$f$. In this case $\cal S$ is realizable, and $G_8$ and $G_{11}$ are drawn as in~\cref{fi:non-minimal-story-8-solution-G8,fi:non-minimal-story-8-solution-G11}.

\begin{figure}[htb]
    \centering
    \includegraphics[page=6, width=0.5\columnwidth]{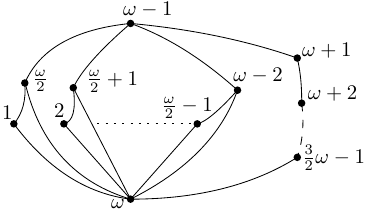}
    \caption{Illustration for \cref{th:sp-non-minimal}. Case $\omega>8$ even.}
    \label{fi:sp-non-minimal-even-G}
\end{figure}

Consider now the case in which $\omega > 8$ and even. Graph $G$ is similar to the one described in the previous case, but it has $\frac{\omega}{2}-1$ parallel flags and a path of length~$\frac{\omega}{2}$, as shown in \cref{fi:sp-non-minimal-even-G}.
%where the vertices are labeled with their subscript in the sequence $\tau = \langle v_1, v_2, \dots, v_{\frac{3}{2}\omega-1} \rangle$.
In any planar embedding of $G_\omega$, at most two vertices of the sequence $v_1,\dots, v_{\frac{\omega}{2}-1}$ can be inside the same cycle $v_\omega, v_i, v_{i+1}, v_{\omega-1}$ ($i=\frac{\omega}{2}, \dots, \omega-2$).
%can be incident to the same face.
%Further, observe that in all the embeddings of $G_\omega$, for $i=1,\dots, \frac{\omega}{2}-1$, either the cycle $v_\omega, v_{\frac{\omega}{2}+i}, v_{\omega-1}, v_{\frac{\omega}{2}+i+1}$ separates $v_{i}$ from $v_{i+1}$ (\cref{fi:G omega even 1l 2l}) or the cycle $v_\omega, v_{\frac{\omega}{2}+i}, v_{\omega-1}, v_{\frac{\omega}{2}+i+1}$ separates or $v_{\omega+i}$ from $v_{\omega+i+1}$ \cref{fi:G omega+i even}, thus the path $v_{\omega+1},\dots, v_{\frac{3}{2}\omega-2}$ can not be drawn without the help of $k$ extra points of $P$ for some $k>0$.
If $k = \frac{\omega}{2}-4$, in $G_\omega$ the number of free points (i.e., points to which no vertex is mapped) is equal to the number of flag components minus $3$. Let $f$ be a face shared by two vertices of the sequence $v_1,\dots, v_{\frac{\omega}{2}-1}$, say $v_1$ and $v_2$, and assume that all the $k$ free points are inside $f$ (see, e.g., \cref{fi:non-minimal-even-failure-omega}). For $G_{\omega+1}, \dots, G_{\frac{3}{2}\omega-2}$ all the vertices of the path can be drawn using the $k$ extra points, as they are incident to the same face. Note that $G_{\omega+1}$ does not contain vertex $v_1$ and $G_{\omega+2}$ does not contain vertex $v_2$, thus face $f$ has been merged with the faces of the flags comprehending $v_1$ and $v_2$, and this new face $f'$ contains the two points to which $v_1$, $v_2$ were mapped. In $G_{\frac{3}{2}\omega-1}$, vertex $v_{\frac{3}{2}\omega-1}$ should be drawn in $f'$ so to maintain planarity, but all the points in $f'$ have been used to draw the $\frac{\omega}{2}-2$ vertices of the path and the $k$ free points are inside other faces (at most two are in the same face). Thus $(G,\omega, k,\tau)$ is not realizable for $k  = \frac{\omega}{2}-4 < \lfloor\frac{\omega}{2}\rfloor-3$ (see, e.g., \cref{fi:non-minimal-even-failure-omega+i}). Note that with~one more point on $f'$, i.e., $k  = \frac{\omega}{2}-3 = \lfloor\frac{\omega}{2}\rfloor-3$, $\cal S$ is realizable; see \cref{fi:non-minimal-even-success-omega,fi:non-minimal-even-success-omega+i}.
\begin{figure}[htb]
    \centering
     \subfigure[$G_\omega$ for $k=\frac{\omega}{2}-4$]{\label{fi:non-minimal-even-failure-omega}\includegraphics[page=7, width=0.45\columnwidth]{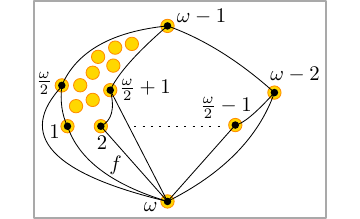}}
     \hfil
     \subfigure[$G_{\frac{3}{2}\omega-2}$ for $k=\frac{\omega}{2}-4$ (failure)]{\label{fi:non-minimal-even-failure-omega+i}\includegraphics[page=8, width=0.45\columnwidth]{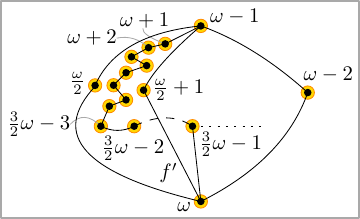}}
     \\
     \subfigure[$G_\omega$ for $k=\frac{\omega}{2}-3$]{\label{fi:non-minimal-even-success-omega}\includegraphics[page=9, width=0.45\columnwidth]{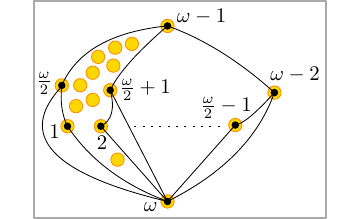}}
     \hfil
     \subfigure[$G_{\frac{3}{2}\omega-2}$ for $k=\frac{\omega}{2}-3$]{\label{fi:non-minimal-even-success-omega+i}\includegraphics[page=10, width=0.45\columnwidth]{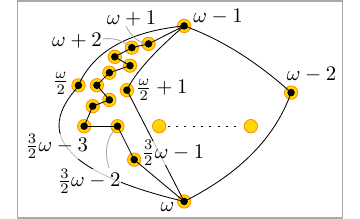}}
    \caption{Illustration for \cref{th:sp-non-minimal}. Case $\omega>8$ even. Drawings of $G_\omega$ and~$G_{\frac{3}{2}\omega-2}$.}
    \label{fi:non-minimal-even}
\end{figure}

Finally, consider the case in which $\omega > 8$ and odd. Again, graph $G$ is similar to the one of previous case but there are $\frac{\omega+1}{2}-1$ parallel flags, two of which are on the same parallel component, thus creating a \emph{double flag}, and the path has length $\frac{\omega+1}{2}$, as shown in \cref{fi:sp-non-minimal-odd-G}.
%where the vertices are labeled with their subscript in the sequence $\tau = \langle v_1, v_2, \dots, v_{\frac{3\omega+1}{2}\omega-1} \rangle$. 
\begin{figure}[htb]
    \centering
    \includegraphics[page=11, width=0.6\columnwidth]{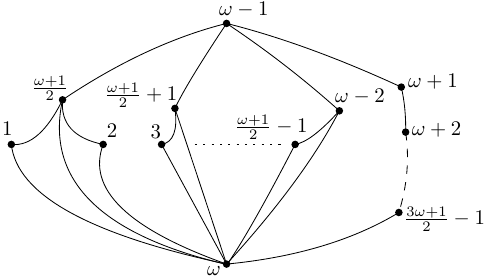}
    \caption{Illustration for \cref{th:sp-non-minimal}. Case $\omega>8$ odd. The leftmost parallel has a double flag.}
    \label{fi:sp-non-minimal-odd-G}
\end{figure}
In any planar embedding of $G_\omega$, at most three vertices of the sequence $v_1,\dots, v_{\frac{\omega+1}{2}-1}$ can be inside the same cycle $v_\omega, v_i, v_{i+1}, v_{\omega-1}$ ($i=\frac{\omega+1}{2}, \dots, \omega-2$), due to the presence of the double flag. Note that at most one face can be shared by more than two vertices of the sequence. If $k = \frac{\omega+1}{2}-5$, in $G_\omega$ there is a number of free points (where no vertex is drawn) equal to the number of flag components minus $4$. Let $f$ be the face shared by two vertices of the sequence $v_1,\dots, v_{\frac{\omega}{2}-1}$, say $v_2$ and $v_3$, one of which is part of the double flag containing also $v_1$, and assume that all the $k$ free points are inside $f$ (see, e.g., \cref{fi:non-minimal-odd-failure-omega}). For $G_{\omega+1}, \dots, G_{\frac{3}{2}\omega-4}$ all the vertices of the path can be drawn using the $k$ extra points, as they are incident to the same face. Note that $G_{\omega+1}$ does not contain vertex $v_1$, $G_{\omega+2}$ does not contain vertex $v_2$, and $G_{\omega+3}$ does not contain vertex $v_3$, thus the face $f$ has been merged with the faces of the flags comprehending $v_1$, $v_2$, and $v_3$, and this new face $f'$ contains the three points to which $v_1$, $v_2$, and $v_3$ were mapped. In $G_{\frac{3\omega +1}{2}-1}$, vertex $v_{\frac{3\omega +1}{2}-1}$ should be drawn in $f'$ so to maintain planarity, but all the points in $f'$ have been used to draw the $\frac{\omega+1}{2}-2$ vertices of the path and the $k$ free points are inside other faces (at most two are in the same face). Thus $(G,\omega, k,\tau)$ is not realizable for $k = \frac{\omega+1}{2}-5 < \lfloor\frac{\omega}{2}\rfloor-3$ (see, e.g., \cref{fi:non-minimal-odd-failure-omega+i}). Note that with one more point~on~$f'$, i.e.,  $k = \frac{\omega+1}{2}-4 = \lfloor\frac{\omega}{2}\rfloor-3$, $\cal S$ is realizable; see  \cref{fi:non-minimal-odd-success-omega,fi:non-minimal-odd-success-omega+i}.
\begin{figure}[htb]
    \centering
    \begin{tabular}{c c}
         \subfigure[$G_\omega$ for $k=\frac{\omega+1}{2}-5$]{\label{fi:non-minimal-odd-failure-omega}\includegraphics[page=12, width=0.45\columnwidth]{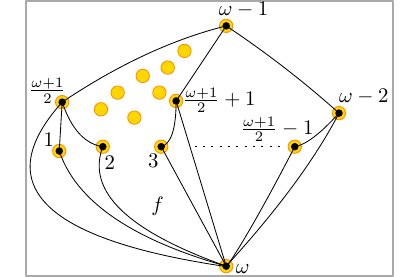}}
         &
         \subfigure[$G_{\frac{3\omega+1}{2}-1}$ for $k=\frac{\omega+1}{2}-5$ (failure)]{\label{fi:non-minimal-odd-failure-omega+i}\includegraphics[page=13, width=0.45\columnwidth]{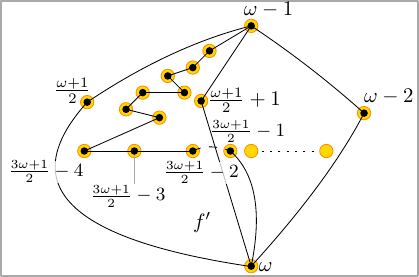}}
         \\
         \subfigure[$G_\omega$ for $k=\frac{\omega+1}{2}-4$]{\label{fi:non-minimal-odd-success-omega}\includegraphics[page=14, width=0.45\columnwidth]{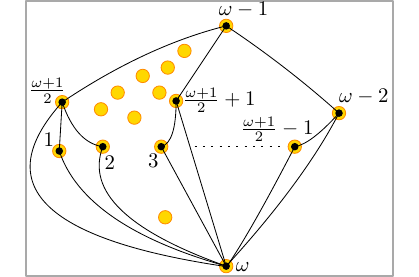}}
         &
         \subfigure[$G_{\frac{3\omega+1}{2}-1}$ for $k=\frac{\omega+1}{2}-4$]{\label{fi:non-minimal-odd-success-omega+i}\includegraphics[page=15, width=0.45\columnwidth]{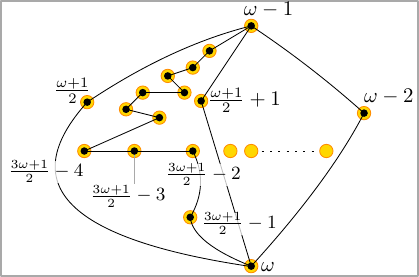}}
    \end{tabular}
    \caption{Illustration for \cref{th:sp-non-minimal}. Case $\omega>8$ odd. Drawings of $G_\omega$ and~$G_{\frac{3\omega+1}{2}-1}$.}
    \label{fi:non-minimal-odd}
\end{figure}
\end{appendixproof}

\section{Final Remarks and Open Problems}
We conclude with some open research directions.
\begin{inparaenum}[$(i)$]
    \item \cref{th:hardness-non-minimal} implies that the realizability testing of graph stories is para\NP-hard when parameterized by~$k$. On the other hand, \cref{th:fpt-non-minimal} proves that the problem is in \FPT\ when parameterized by $\omega+k$. For non-minimal graph stories, it remains open to establish the complexity of the realizability problem when parameterized by~$\omega$ alone.
    %
    %It would be interesting to study the realizability problem for stories whose graphs are not planar.
    %
    \item About minimal graph stories, we showed that for $\omega \geq 5$ there are stories of series-parallel graphs that are not realizable. For $k=1$, the smaller $\omega$ for which we have a non-realizable story of a series-parallel graph is $10$. What about the realizability of series-parallel graphs for $k =1$ and $5 \leq \omega \leq 9$?
    \item Finally, is any (minimal) graph story $h$-reroute realizable for $h$ being a constant or a sublinear function of $\omega$? 
\end{inparaenum}

%\clearpage

\bibliographystyle{splncs04}
\bibliography{biblio}

\end{document}